%% file: main.tex
\documentclass[review]{elsarticle}

\usepackage{lineno}
\usepackage{hyperref}
\modulolinenumbers[5]

\RequirePackage{amsmath}
\RequirePackage{amsfonts}
\RequirePackage{amssymb}
\RequirePackage{bbm}

\journal{Computer Methods in Applied Mechanics and Engineering}

\usepackage{siunitx}
\usepackage{multirow}

\usepackage{subcaption}
\usepackage{nicefrac}
\usepackage{amsthm,thmtools}

\usepackage[margin=2.5cm]{geometry}
\DeclareSymbolFont{Symbols}{OMS}{cmsy}{m}{n}
\SetSymbolFont{Symbols}{bold}{OMS}{cmsy}{b}{n}
\DeclareMathSymbol{\Setminus}{\mathbin}{Symbols}{"6E}
\usepackage{MnSymbol}

\usepackage{empheq}
\usepackage{pgfplots}
\pgfplotsset{compat=newest}
\usepackage{pgfplotstable}
\usepackage{colortbl}
\usepackage{tikz}
\usetikzlibrary{arrows,cd,external,calc,matrix,positioning}
\hypersetup{hidelinks}

\newtheorem{remark}{Remark}

\newtheorem{theorem}{Theorem}

\newtheorem{lemma}[theorem]{Lemma}

\newtheorem{corollary}[theorem]{Corollary}

\newcommand{\closure}[1]{\ensuremath{\overline{#1}}}

\newcommand{\cellSet}[1]{\ensuremath{\mathcal{C}_{\mathrm{#1}}}}
\newcommand{\faceSet}[1]{\ensuremath{\mathcal{F}_{\mathrm{#1}}}}
\newcommand{\edgeSet}[1]{\ensuremath{\mathcal{E}_{\mathrm{#1}}}}
\newcommand{\nodeSet}[1]{\ensuremath{\mathcal{N}_{\mathrm{#1}}}}
\newcommand{\graph}[1]{\ensuremath{\mathcal{G}_{\mathrm{#1}}}}
\newcommand{\crossEdges}{\ensuremath{\edgeSet{II}}}
\newcommand{\dimSpace}[1]{\ensuremath{\revised{\operatorname{dim}}\left(#1\right)}}
\newcommand{\cardSet}[1]{\ensuremath{\revised{\operatorname{card}}\left(#1\right)}}

\newcommand{\Curl}{\ensuremath{\operatorname{curl}}}
\newcommand{\Div}{\ensuremath{\operatorname{div}}}
\newcommand{\Grad}{\ensuremath{\operatorname{grad}}}
\newcommand{\trans}{\ensuremath{^{\top}}}

\newcommand{\reluctivity} {\ensuremath{\nu}} 
\newcommand{\vecs}[1]{\boldsymbol{#1}}

\newcommand{\vfield}{\ensuremath{\vecs{v}}}
\newcommand{\baseComp}{\ensuremath{\vecs{w}}}
\newcommand{\Afield}{\ensuremath{\vecs{A}}}
\newcommand{\Jfield}{\ensuremath{\vecs{J}}} 
\newcommand{\Bfield}{\ensuremath{\vecs{B}}}
\newcommand{\nfield}{\ensuremath{\vecs{n}}}
\newcommand{\zfield}{\ensuremath{\vecs{0}}}

\newcommand{\harmonic}{\ensuremath{\vecs{h}}}
\newcommand{\gNfield}{\ensuremath{\vecs{g}_{\mathrm{N}}}}
\newcommand{\gDfield}{\ensuremath{\vecs{g}_{\mathrm{D}}}}

\newcommand{\Hcurl}[1]{\ensuremath{{H}(\mathrm{curl};#1)}}
\newcommand{\Hncurl}[1]{\revised{\ensuremath{{H}_{0,\Gamma_{\mathrm{D}}}(\mathrm{curl};#1)}}}
\newcommand{\norm}[2]{\ensuremath{\left\vert\left\vert #1 \right\vert\right\vert_{#2}}}
\newcommand{\volPairing}[2]{\ensuremath{\int_{#1}#2\operatorname{dV}}}
\newcommand{\surfPairing}[2]{\ensuremath{\int_{#1}#2\operatorname{dA}}}
\newcommand{\errBfield}{\ensuremath{\epsilon_{\Bfield}}}

\newcommand{\stiffMat}{\ensuremath{\mathbf{K}}}
\newcommand{\couplMat}{\ensuremath{\mathbf{B}}}
\newcommand{\primMat}{\ensuremath{\mathbf{C}_{\mathrm{p}}}}
\newcommand{\coarseMat}{\ensuremath{\mathbf{F}}}
\newcommand{\interfaceMat}{\ensuremath{\mathbf{S}}}
\newcommand{\dofs}{\ensuremath{\mathbf{a}}}
\newcommand{\mults}{\boldsymbol{\lambda}}
\newcommand{\rhs}{\ensuremath{\mathbf{j}}}
\newcommand{\cond}[1]{\ensuremath{\kappa\big(#1\big)}}
\newcommand{\minEigv}[1]{\ensuremath{\sigma_{\mathrm{min}}\big(#1\big)}}
\newcommand{\maxEigv}[1]{\ensuremath{\sigma_{\mathrm{max}}\big(#1\big)}}
\newcommand{\nCoarse}{\revised{\ensuremath{n_{\mathrm{gp}}}}}

\newcommand{\cgTol}{\ensuremath{\varepsilon_{\mathrm{tol}}}}

\newcommand{\revised}[1]{{\color{black}#1}}
\newcommand{\ma}[1]{{\color{black}#1}}

\usepackage[exponent-product=\cdot]{siunitx}

\begin{document}
    
	\begin{frontmatter}

		\title{Tree-Cotree-Based Tearing and Interconnecting for 3D Magnetostatics:\\ A Dual-Primal Approach}

		\author[cemaddress,cceaddress]{Mario Mally}
		\ead{mario.mally@tu-darmstadt.de}

		\author[epfladdress]{Bernard Kapidani}
		\ead{bernard.kapidani@epfl.ch}

		\author[cemaddress,cceaddress]{Melina Merkel}
		\ead{melina.merkel@tu-darmstadt.de}

		\author[cemaddress,cceaddress]{Sebastian Schöps}
		\ead{sebastian.schoep@tu-darmstadt.de}

		\author[uscaddress,citmaga]{Rafael Vázquez}
		\ead{rafael.vazquez@usc.es}

		\address[cemaddress]{Computational Electromagnetics Group, Technische Universität Darmstadt, 64289 Darmstadt, Germany}
		\address[cceaddress]{Centre for Computational Engineering, Technische Universität Darmstadt, 64289 Darmstadt, Germany}
		\address[epfladdress]{Chair of Numerical Modelling and Simulation, École Polytechnique Fédérale de Lausanne, 1015 Lausanne, Switzerland}
		\address[uscaddress]{Department of Applied Mathematics, Universidade de Santiago de Compostela, 15782, Santiago de Compostela, Spain}
		\address[citmaga]{Galician Centre for Mathematical Research and Technology (CITMAga), 15782, Santiago de Compostela, Spain}
		\begin{abstract}
			The simulation of electromagnetic devices with complex geometries and large-scale discrete systems benefits from advanced computational methods like IsoGeometric Analysis and Domain Decomposition. In this paper, we employ both concepts in an Isogeometric Tearing and Interconnecting method to enable the use of parallel computations for magnetostatic problems. We address the underlying non-uniqueness by using a graph-theoretic approach, the tree-cotree decomposition. The classical tree-cotree gauging is adapted to be feasible for parallelization, which requires that all local subsystems are uniquely solvable.
			Our contribution consists of an explicit algorithm for constructing compatible trees and combining it with a dual-primal approach to enable parallelization. The correctness of the proposed approach is proved and verified by numerical experiments, showing its accuracy, scalability and optimal convergence.
		\end{abstract}

		\begin{keyword}
			 IGA\sep IETI\sep Dual-Primal\sep Tree-Cotree Gauging\sep Magnetostatics
		\end{keyword}

	\end{frontmatter}

	\section{Introduction}
    \input{introduction}

	\section{Magnetostatics on Decomposed Domains}\label{sec:form}
    \input{formulations}

    \section{Tree-Cotree Gauging for Tearing and Interconnecting}\label{sec:treeGraph}
    \input{treeCotree}

    \section{Dual-Primal Formulation for 3D Magnetostatics}\label{sec:dualPrimal}
    \input{dpForm}
    
    \section{Numerical Experiments}\label{sec:nums}
    \input{numerics}
    
    \section{Conclusions and Outlook}
    \input{conclusion}

    \section*{Acknowledgements}
    The work is supported by the joint DFG/FWF Collaborative Research Centre CREATOR (DFG: Project-ID 492661287/TRR 361; FWF: 10.55776/F90) at TU Darmstadt, TU Graz and JKU Linz. We acknowledge the funding of The "Ernst Ludwig Mobility Grant" of the Association of Friends of Technical University of Darmstadt e.V.

    \section*{Declaration of Competing Interest}
    The authors declare that they have no known competing financial interests or personal relationships that could have appeared to influence the work reported in this paper.

    \section*{Declaration of Generative AI and AI-Assisted Technologies in the Writing Process}
    During the preparation of this work the authors used \texttt{DeepL} and \texttt{ChatGPT} in order to improve the readability and language of this paper. After using these tools/services, the authors reviewed and edited the content as needed and take full responsibility for the content of the publication.

    \appendix
    \input{appendix}

    \bibliographystyle{abbrvnat}
    \bibliography{bibtex}

\end{document}

%% file: introduction.tex
Simulation-based design workflows for engineering products are ubiquitous in academia and industry. Their success relies on accurate and efficient solvers. A suitable method to achieve this is IsoGeometric Analysis (IGA), due to its ability to accurately represent complex geometries and to provide higher continuity in the solution space. This is achieved by using NURBS (Non-Uniform Rational B-Splines) to represent the geometry and \revised{NURBS or} B-Splines to approximate the solution \cite{Hughes_2005aa}. These advantages are particularly attractive for the magnetostatic simulation of permanent magnet synchronous machines due to their cylindrical shape \cite{Bontinck_2018ac} and the computation of sensitive quantities, such as torque \cite{Merkel_2022ab}. This is made possible by the introduction of spline spaces suitable for the representation of electromagnetic fields \cite{Buffa_2010aa,Buffa_2015aa,Buffa_2019ac}.

However, as the complexity and size of simulation models increase, it becomes essential to exploit parallel computing architectures to effectively handle the computational load. We follow the Domain Decomposition (DD) paradigm \cite{Toselli_2005aa} and in particular the Tearing and Interconnecting (TI) methods originally proposed by \citeauthor{Farhat_1991aa} in \cite{Farhat_1991aa}. These methods are probably best known as Finite Element Tearing and Interconnecting (FETI), while the IGA counterpart is called Isogeometric Tearing and Interconnecting (IETI) \cite{Kleiss_2012ab,Bouclier_2022aa}. Both decompose the computational domain into subdomains, couple them weakly at conformal interfaces using Lagrange multipliers, and thus promise scalability for large problems. However, solvability is a fundamental issue of decomposed Laplace-type problems due to the lack of essential boundary conditions for `floating' subdomains. Two main approaches address this problem: The dual-primal approach, see e.g. \cite{Farhat_2001aa,Klawonn_2002aa,Kleiss_2012ab} or the total/all-floating approach~\cite{Dostal_2006aa,Of_2009ab,Bosy_2020aa}. 

For high-frequency electromagnetic problems, FETI has been investigated for scattering and wave propagation, as discussed in \cite{Li_2009ab} and references therein. In the case of magnetism, it was applied to the scalar potential formulation, i.e. a Laplace-type problem, in \cite{Marcsa_2013aa,Ghenai_2024aa,Schneckenleitner_2022aa}. The vector potential formulation suffers from non-uniqueness issues due to the kernel of the curl operator. This formulation is considered in \cite{Toselli_2006aa,Keranen_2015aa} and \cite{Yao_2012aa,Yao_2012ab}. The first two contributions regularize the problem by perturbing the operator. We follow and extend the idea of \citeauthor{Yao_2012aa} \revised{in \cite{Yao_2012aa,Yao_2012ab}} and employ a tree-cotree decomposition \cite{Albanese_1988aa,Munteanu_2002aa,Kapidani_2022aa,Rapetti_2022aa}.

Let us state our contributions. We provide a mathematical justification of the approach from \citeauthor{Yao_2012aa} \revised{Their work only considers} problems with full Dirichlet boundary conditions and low-order finite elements. We give a new explicit algorithm for constructing a tree-cotree decomposition that also applies to problems with mixed boundary conditions and high-order isogeometric discretizations. Finally, we validate our approach with numerical experiments using the IGA-library \texttt{GeoPDEs} \cite{Vazquez_2016aa}.

The paper is structured as follows: \autoref{sec:form} introduces the relevant Maxwell equations and provides an overview of the weak FETI/IETI formulation for magnetostatics as well as its discretization. \autoref{sec:treeGraph} explains the fundamental concepts of tree-cotree decompositions and presents an explicit algorithm for constructing appropriate trees.
In \autoref{sec:dualPrimal}, we examine and analyze the constructed tree within the context of dual-primal TI, detailing how to correctly select primal degrees of freedom to ensure local invertibility in general settings. Finally, in \autoref{sec:nums} the proposed approach is verified with numerical experiments and its scalability is examined with respect to relevant properties for large-scale problems.

%% file: formulations.tex
In the low-frequency regime, e.g. for the simulation of permanent magnet synchronous machines \cite{Salon_1995aa}, the magnetostatic approximation of Maxwell's equations \cite{Jackson_1998aa}
\begin{empheq}[left = \empheqlbrace]{align}
	\Curl\left(\reluctivity\Bfield\right) &= \Jfield\label{eq:magnstat1}\\
	\Div\Bfield &= 0\label{eq:magnstat2}
\end{empheq}
is most relevant. The system is formulated in terms of the magnetic flux density $\Bfield$ and the excitation current density $\Jfield$. For the underlying material parameter $\reluctivity$ (reluctivity), one typically assumes \revised{$0<\reluctivity_{1}\leq\reluctivity\leq\reluctivity_{2}<\infty$}. \revised{We introduce a magnetic vector potential via $\Bfield=\Curl\Afield$, which automatically satisfies~\eqref{eq:magnstat2}. Inserting it into~\eqref{eq:magnstat1} results in 
\begin{align}
	\Curl\left(\reluctivity\Curl\Afield\right)&=\Jfield.
	\label{eq:vecPot1}
\end{align}
We assume that these laws are prescribed on a bounded, simply-connected domain $\Omega\subset\mathbb{R}^{3}$ which we will relax only in \autoref{subsec:toroidal} where a hollow cylinder is simulated. On the boundary $\partial\Omega$, conditions are prescribed to account for the behavior of the fields outside of the domain. The most common conditions are
\begin{align}
    \Afield\times\nfield &= \gDfield \quad\text{on}~~\Gamma_{\mathrm{D}},\label{eq:vecPot2}\\
    \left(\reluctivity\Curl\Afield\right)\times\nfield &= \gNfield \quad\text{on}~~\Gamma_{\mathrm{N}},\label{eq:vecPot3}
\end{align}
which define (inhomogeneous) Dirichlet and Neumann boundary conditions on $\Gamma_{\mathrm{D}}$ and $\Gamma_{\mathrm{N}}$, respectively. Here, $\nfield$ denotes the outward normal vector of $\partial\Omega$. For simplicity of presentation, we focus on $\gDfield=0$ and $\gNfield=0$ which correspond to perfect electric and magnetic conducting boundary conditions, respectively. We assume that $\partial\Omega=\closure{\Gamma}_{\mathrm{D}}\cup\closure{\Gamma}_{\mathrm{N}}$ (where $\closure{D}$ denotes the closure of open set $D$) and $\Gamma_{\mathrm{D}}\cap\Gamma_{\mathrm{N}}=\emptyset$. Furthermore we require that ${\Gamma}_{\mathrm{D}}$ is connected, which will be relaxed in the test configuration of \autoref{subsec:scalab}.}

The vector potential formulation \eqref{eq:vecPot1}-\eqref{eq:vecPot3} has a major intricacy originating from the non-trivial kernel of the $\mathrm{curl}$-operator. To summarize the descriptions of \cite{Arnold_2018aa}, we can write\revised{
\begin{equation}
    \operatorname{ker}(\Curl) = \operatorname{im}(\Grad)\cup\operatorname{span}\left(\harmonic^{(k)}\right)_{k=1}^{n_{\mathrm{h}}}\label{eq:kernel}
\end{equation}
and characterize the kernel with gradient fields $\operatorname{im}(\Grad)$ and a linear combination of $n_{\mathrm{h}}$ harmonic fields $\harmonic^{(k)}$. \ma{The number of harmonic fields $n_{\mathrm{h}}$ depends on both the number of holes, i.e., the first Betti number, and the configuration of the boundary conditions. In the following, we assume $n_{\mathrm{h}}=0$, i.e., simply-connected domains $\Omega$ with connected Dirichlet boundary $\Gamma_{\mathrm{D}}$.} Nonetheless, we already introduce harmonic fields here, to be able to explain and comment on extensions to this paper. The identity \eqref{eq:kernel} implies that only $\Afield_{\perp}\in\operatorname{ker}(\Curl)^{\perp}$ can be uniquely determined from \eqref{eq:vecPot1}-\eqref{eq:vecPot3}. The remaining part $\mathbf{A}_0\in\operatorname{ker}(\Curl)$ is undetermined and can be chosen freely, due to}
\begin{equation}
    \Bfield = \Curl\Afield = \Curl\left(\Afield_{\perp} + \Afield_0\right) = \Curl\Afield_{\perp}.
\end{equation}
By fixing the kernel component $\Afield_0$, which does not affect our quantity of interest $\Bfield$, and reformulating \eqref{eq:vecPot1}, one can compute a unique $\Afield_{\perp}$ if the current density is compatible, i.e. $\Jfield\in\operatorname{ker}(\Curl)^{\perp}$, \revised{see e.g.} \cite{Arnold_2018aa}. This procedure is denoted as gauging in the literature, e.g. \cite{Kettunen_1999aa}. One particular example, for simply-connected domains, is the Coulomb gauge where $\Div\Afield=\Div\Afield_0=0$ is imposed as an additional constraint. In this paper, we focus on the tree-cotree gauge, which is applied in the discrete setting for the TI approach. Consequently, we need to discuss our weak formulation, its discretization, and the extension to domain decomposition before we can elaborate more on the gauging technique.

\subsection{Weak Formulation and its Discretization}\label{subsec:SingleDomDisc}
Let us state \eqref{eq:vecPot1}-\eqref{eq:vecPot3} in its weak form as: Find $\Afield\in\Hncurl{\Omega}$ such that
\begin{equation}
    \volPairing{\Omega}{(\reluctivity\Curl\Afield)\cdot(\Curl\vfield)}=\volPairing{\Omega}{\Jfield\cdot\vfield}+\surfPairing{\Gamma_{\mathrm{N}}}{\gNfield\cdot\vfield}\label{eq:weakVecPot}
\end{equation}
for all $\vfield\in\Hncurl{\Omega}$, where $\Hncurl{\Omega}$ denotes the adequate Hilbert space for vector potentials with vanishing tangential component on $\Gamma_{\mathrm{D}}$ \cite{Monk_2003aa}. We follow the typical Ritz-Galerkin procedure and introduce a finite-dimensional basis $\{\baseComp_i\}_{i=1}^{n_{\mathrm{g}}}$ such that $\mathbb{V}\coloneqq\operatorname{span}(\baseComp_i)_{i=1}^{n_{\mathrm{g}}}\subset\Hcurl{\Omega}$.

We approximate $\Afield$ by a linear combination of finitely many basis functions, i.e.
\begin{equation}
    \Afield_h=\sum_{i=1}^{n_{\mathrm{g}}}a_i\baseComp_i\in\mathbb{V}.
\end{equation}
In this context, we call $a_i$ a degree of freedom (DOF) and introduce $\tilde{\dofs}_{\mathrm{g}}\in\mathbb{R}^{n_{\mathrm{g}}}$ as the vector of all degrees of freedom (DOFs). Inserting the discrete field $\Afield_h$ into \eqref{eq:weakVecPot} leads to
\begin{equation}
    \left(\tilde{\stiffMat}_{\mathrm{g}}\right)_{ij} \coloneqq \volPairing{\Omega}{\reluctivity\left(\Curl\baseComp_j\right)\cdot\left(\Curl\baseComp_i\right)}
\end{equation}
for the entries of stiffness matrix $\tilde{\stiffMat}_{\mathrm{g}}\in\mathbb{R}^{n_{\mathrm{g}}\times n_{\mathrm{g}}}$ and to
\begin{equation}
    \left(\tilde{\rhs}_{\mathrm{g}}\right)_i \coloneqq \volPairing{\Omega}{\Jfield\cdot\baseComp_i} + \surfPairing{\Gamma_{\mathrm{N}}}{\gNfield\cdot\baseComp_i}
\end{equation}
for the right hand side $\tilde{\rhs}_{\mathrm{g}}\in\mathbb{R}^{n_{\mathrm{g}}}$. Here, we added the subscript $\mathrm{g}$ to emphasize that this is the global problem. \ma{Under mild assumptions, it is possible to eliminate the Dirichlet DOFs as in \cite[A.2]{Formaggia_2012aa} to obtain the linear system}
\begin{equation}
    \stiffMat_{\mathrm{g}}\dofs_{\mathrm{g}}=\rhs_{\mathrm{g}}.\label{eq:discVecPot}
\end{equation}
\revised{The entities $\stiffMat_{\mathrm{g}}$, $\rhs_{\mathrm{g}}$ and $\dofs_{\mathrm{g}}$ are the appropriately modified versions of $\tilde{\stiffMat}_{\mathrm{g}}$, $\tilde{\rhs}_{\mathrm{g}}$ and $\tilde{\dofs}_{\mathrm{g}}$, respectively.}

\revised{Note that while we expect our findings to apply to any higher-order FEM discretization on, e.g. tetrahedral meshes~\cite{Alonso-Rodriguez_2020aa,Rapetti_2022aa}, we focus on the multipatch IGA approach with tensor product, high-order, edge-element basis functions. In this framework, the domain $\Omega$ is divided into $N$ patches $\Omega^{(i)}$. The decomposition satisfies $\closure{\Omega}=\bigcup_{i}\closure{\Omega}^{(i)}$ and the subdomains are mutually disjoint, i.e. $\Omega^{(j)}\cap\Omega^{(k)}=\emptyset$ for $j\neq k$. The discrete space is defined as $\mathbb{V} = \{ \ensuremath{\vecs{u}} \in \Hcurl{\Omega} : \ensuremath{\vecs{u}} |_{\Omega^{(i)}} \in \mathbb{V}^{(i)} \}$. For each patch the subspace $\mathbb{V}^{(i)}=S^1_{p}(\Omega^{(i)})$ is a space of splines of mixed degree, and it has the form
\begin{align}
    S_p^1(\Omega^{(i)}) &= \{ \ensuremath{\vecs{u}} \in \Hcurl{\Omega^{(i)}} : \ensuremath{\vecs{u}} \circ \ensuremath{\vecs{F}}^{(i)} = (D \ensuremath{\vecs{F}}^{(i)})^{-\top} \hat {\ensuremath{\vecs{u}}}, \, \hat {\ensuremath{\vecs{u}}} \in S_p^1(\hat \Omega) \}, \label{eq:iga_space1}
    \\
    S_p^1(\hat \Omega) &= S_{p-1,p,p}(\Xi'_1,\Xi_2,\Xi_3) \times S_{p,p-1,p}(\Xi_1,\Xi'_2,\Xi_3) \times S_{p,p,p-1}(\Xi_1,\Xi_2,\Xi'_3),\label{eq:iga_space2}
\end{align}
where $\ensuremath{\vecs{F}}^{(i)}$ is a regular mapping from the unit cube $\hat \Omega$ to the physical domain $\Omega^{(i)}$ and $D \ensuremath{\vecs{F}}^{(i)}$ its non-singular Jacobian, $p$ is the degree, and $\Xi_k$, $\Xi'_k$ are the knot vectors in each parametric direction for degrees $p$ and $p-1$, respectively, see \cite[Sec.~5]{Beirao-da-Veiga_2014aa} for details. As standard in the IGA setting, inserting knots in the knot vectors translates into mesh refinement, and we will denote the mesh size of the patches by $h$.

\ma{We also assume, following \cite[Assumption~3.2]{Beirao-da-Veiga_2014aa}, that the Dirichlet boundary $\Gamma_{\mathrm{D}}$ is a union of patch-faces, i.e., images of sides of the unit cube mapped through $\ensuremath{\vecs{F}}^{(i)}$. This simplifies the elimination of Dirichlet DOFs, see \cite[Sec.~4.1.3]{Vazquez_2016aa}.}


With the choice of the basis as in \cite{Beirao-da-Veiga_2014aa}, there is a one-to-one correspondence between the basis functions $\baseComp_i$ (or the DOFs $a_i$) and the edges of the control mesh \cite{Kapidani_2022aa,Kapidani_2024aa}, which we interpret as a graph and denote by
\begin{equation*}
    \graph{}=\left(\nodeSet{},\edgeSet{}\right).
\end{equation*}
We write $\nodeSet{}$ and $\edgeSet{}$ for the sets of nodes and edges of the graph, respectively. In the lowest-order case, i.e. $p=1$, this control mesh corresponds to the mesh of a conventional FEM discretization with lowest-order Nédélec-type hexahedral elements.

\begin{remark}
    We note that the Dirichlet boundary condition could be directly incorporated into the spaces $\mathbb{V}$ and $\mathbb{V}^{(i)}$. However, we are not eliminating the corresponding DOFs, because it helps to identify the basis functions with edges of the control mesh. Maintaining the functions on the Dirichlet boundary also simplifies the explanation of how to construct an appropriate spanning tree.
\end{remark}
}

\subsection{Nonoverlapping Domain Decomposition}
\revised{Although not the only choice, it is natural to consider each patch $\Omega^{(i)}$ of the multipatch decomposition \eqref{eq:iga_space2} as a separate subdomain in the TI context, see e.g. \cite{Kleiss_2012ab}. Consequently, we have $N$ subdomains for an IGA multipatch space consisting of $N$ patches. This decomposition automatically fulfills all usual assumptions including simply-connected subdomains since each patch is given by a regular mapping from the unit cube, see \eqref{eq:iga_space1}. Consequently, we also consider
\begin{equation*}
    \mathbb{V}^{(i)}=S_p^1(\Omega^{(i)})=\operatorname{span}\left(\baseComp_j^{(i)}\right)_{j=1}^{n_{i}}\subset\Hcurl{\Omega^{(i)}}
\end{equation*}
for the discretization of a local magnetostatic problem on each subdomain $\Omega^{(i)}$, where $\{ \baseComp_j^{(i)}\}_{j=1}^{n_i}$ denote the basis functions of $\mathbb{V}^{(i)}$.

Let us introduce some more notation related to domain decomposition.} Each boundary $\partial\Omega^{(i)}$ shall consist of open facets $\Gamma^{(i,k)}$ such that $\partial\Omega^{(i)}=\bigcup_k\closure{\Gamma}^{(i,k)}$ and $\Gamma^{(i,k)}\cap\Gamma^{(i,\ell)}=\emptyset$ for $k\neq\ell$. \revised{A facet is either a subset of the Dirichlet or Neumann boundary, i.e. $\Gamma^{(i,k)}\subset\Gamma_{\mathrm{D}}$ or $\Gamma^{(i,k)}\subset\Gamma_{\mathrm{N}}$, or an interface $\Gamma_{ij}:=\Gamma^{(i,k)}=\Gamma^{(j,l)}$ for some $i<j$ and $k$,$l$.} Finally, each facet boundary consists of edges $\closure{\lambda}^{(i,j,k)}$ such that $\partial\Gamma^{(i,j)}=\bigcup_{k}\closure{\lambda}^{(i,j,k)}$ and $\lambda^{(i,j,k)}\cap\lambda^{(i,j,\ell)}=\emptyset$ for $k\neq\ell$. The union of all edges \revised{$\closure{\lambda}^{(i,j,k)}$} (including their endpoints) defines the wire basket $\mathcal{W}^{(i)}:=\bigcup_{j,k}\closure{\lambda}^{(i,j,k)}$, which is visualized in \autoref{fig:subGs_a}.

For the local (control) meshes, we introduce graphs again. Analogous to the previous notation, $\graph{}^{(i)} = \left(\nodeSet{}^{(i)},\edgeSet{}^{(i)}\right)$ denotes the (sub)graph formed by the nodes $\nodeSet{}^{(i)}$ and edges $\edgeSet{}^{(i)}$ related to the basis functions in \revised{$\mathbb{V}^{(i)}$}. Furthermore, $\graph{\Gamma}^{(i)} = \left(\nodeSet{\Gamma}^{(i)},\edgeSet{\Gamma}^{(i)}\right)$ and $\graph{\lambda}^{(i)} = \left(\nodeSet{\lambda}^{(i)},\edgeSet{\lambda}^{(i)}\right)$ are the graphs for the basis functions on the boundaries and local wire baskets, respectively. Note that these subdomain-wise definitions can be extended to the global domain, i.e. $\graph{\Gamma}$ and $\graph{\lambda}$, by creating unions due to the consistency of the decomposition. A visualization of these local and global subgraphs is provided in \autoref{fig:subGs_b} and \autoref{fig:subGs_c}.

\begin{figure}
    \begin{subfigure}[b]{0.33\textwidth}
        \centering
        \includegraphics[keepaspectratio]{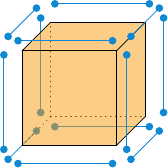}
        \caption{Local wire basket.}
        \label{fig:subGs_a}
    \end{subfigure}
    \begin{subfigure}[b]{0.33\textwidth}
        \centering
        \includegraphics[keepaspectratio]{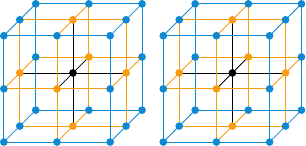}
        \caption{Local subgraphs.}
        \label{fig:subGs_b}
    \end{subfigure}
    \begin{subfigure}[b]{0.33\textwidth}
        \centering
        \includegraphics[keepaspectratio]{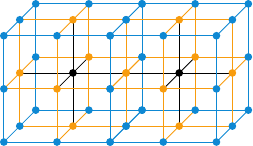}
        \caption{Global subgraphs.}
        \label{fig:subGs_c}
    \end{subfigure}
    \caption{Visualization of wire basket (blue) and $\partial\Omega^{(i)}$ (orange) for one hexahedral subdomain in \autoref{fig:subGs_a}. In \autoref{fig:subGs_b} and \autoref{fig:subGs_c}, the corresponding subgraphs for two connected subdomains are shown. There, blue indicates $\graph{\lambda}^{(i)}$ and $\graph{\lambda}$. Furthermore, orange together with blue indicates $\graph{\Gamma}^{(i)}$ and $\graph{\Gamma}$. The inner black nodes and edges belong to neither wire basket nor boundary graphs.}
    \label{fig:subGs}
\end{figure}

\subsubsection{The Local Discrete Problems and Their Coupling}
Introducing discrete, local vector potentials $\Afield^{(i)}_h \in \mathbb{V}^{(i)}$ on every subdomain $\Omega^{(i)}$ with $\Afield^{(i)}_h=\sum_{j=1}^{n_i}a^{(i)}_j\baseComp^{(i)}_j$ yields the uncoupled system of independent problems
\begin{equation}
    \underbrace{\begin{bmatrix}
        \tilde{\stiffMat}^{(1)} & & \\
        & \ddots & \\
        & & \tilde{\stiffMat}^{(N)} \\
    \end{bmatrix}}_{\eqqcolon\tilde{\stiffMat}}\underbrace{\begin{bmatrix}
        \tilde{\dofs}^{(1)} \\
        \vdots \\
        \tilde{\dofs}^{(N)} \\
    \end{bmatrix}}_{\eqqcolon\tilde{\dofs}}=\underbrace{\begin{bmatrix}
        \tilde{\rhs}^{(1)} \\
        \vdots \\
        \tilde{\rhs}^{(N)} \\
    \end{bmatrix}}_{\eqqcolon\tilde{\rhs}}\label{eq:locProbs}
\end{equation}
when applying the principles behind \eqref{eq:weakVecPot}-\eqref{eq:discVecPot} to every subdomain. \ma{Note that the linear system \eqref{eq:locProbs} only contains local information. In order to compute solutions that also solve \eqref{eq:discVecPot}, it is essential to provide transmission conditions (also called coupling constraints), see \cite{Quarteroni_1999aa}, in addition to prescribing the Dirichlet conditions.}
Due to the conforming decomposition of $\Omega$ and of \revised{the} space $\mathbb{V}$, \revised{the necessary} coupling constraints are defined as
\begin{equation}
    \tilde{\couplMat}\tilde{\dofs}=\zfield,
\end{equation}
where $\tilde{\couplMat}\in\{-1,0,1\}^{m\times n}$ is a signed Boolean matrix. Here, $n$ is given by $n=\sum_{i=1}^{N}n_i$ and each row of $\tilde{\couplMat}$ encodes a constraint of the form
\begin{equation}
    a_{i_j}^{(j)} - a_{i_k}^{(k)} = 0, \label{eq:coupling}
\end{equation}
which enforces tangential continuity of the discrete solutions $\Afield^{(i)}_h$ on all interfaces.
The \revised{arising $m$} constraints are then incorporated by using an approach based on Lagrange multipliers which we denote here as $\tilde{\mults}\in\mathbb{R}^{m}$. This results in
\begin{equation}
    \begin{bmatrix}
        \stiffMat & \couplMat\trans \\
        \couplMat & \mathbf{0} \\
    \end{bmatrix}\begin{bmatrix}
        \dofs \\
        \mults \\
    \end{bmatrix}=\begin{bmatrix}
        \rhs \\
        \mathbf{0} \\
    \end{bmatrix},\label{eq:TIproblem}
\end{equation}
where the Dirichlet boundary conditions are already imposed. The underlying \revised{elimination of Dirichlet DOFs, \cite[A.2]{Formaggia_2012aa}, affects the local systems again, i.e. $\stiffMat$, $\rhs$ and $\dofs$ are the modified versions of $\tilde{\stiffMat}$, $\tilde{\rhs}$ and $\tilde{\dofs}$ respectively. Additionally also the coupling constraints need to be modified.} At the intersection between the Dirichlet boundary parts and interfaces, the coupling \revised{\eqref{eq:coupling} between interface DOFs produces zero rows in the coupling matrix $\tilde{\couplMat}$} when Dirichlet DOFs are eliminated from the system. After removing the corresponding constraints and multipliers, we write $\couplMat$ and $\mults$ instead of $\tilde{\couplMat}$ and $\tilde{\mults}$, respectively.

%% file: treeCotree.tex
\revised{When considering the global problem \eqref{eq:discVecPot} typical gauging approaches, such as the tree-cotree gauging, can be straightforwardly applied to handle $\dimSpace{\ker\left(\stiffMat_{\mathrm{g}}\right)}>0$, which originates from \eqref{eq:kernel} and our choice of~$\mathbb{V}$ as explained in \cite{Kapidani_2022aa}. The difficulties arise if one wants to compute solutions to the local problems independently which requires gauging the matrix $\stiffMat$ from \eqref{eq:TIproblem}. The reason is that 
\begin{equation}
    \dimSpace{\ker\left(\stiffMat_{\mathrm{g}}\right)}\leq\dimSpace{\ker\left(\stiffMat\right)}\label{eq:gloLocKernel}
\end{equation}
holds. To gain insight into this relationship, one can draw an analogy with the Laplace problem. There, a constant $c_{\mathrm{g}}$ remains unspecified by the discrete problem with $\partial\Omega=\Gamma_{\mathrm{N}}$, while every subdomain block leaves one constant $c_{i}$ unspecified by itself. Due to the coupling constraints one implicitly enforces $c_{\mathrm{g}}=c_{1}=\ldots=c_{N}$ which implies that every local $c_{i}$ cannot be chosen freely but needs to be consistent with the global problem. Note that Dirichlet boundary conditions automatically fix $c_{\mathrm{g}}=c_{j}$ with $\partial\Omega^{(j)}\cap\Gamma_{\mathrm{D}}\neq\emptyset$ but that $c_{k}$ remains unspecified in the local context for ``floating'' subdomains with $\partial\Omega^{(k)}\cap\Gamma_{\mathrm{D}}=\emptyset$. In the $\mathrm{curl}$-$\mathrm{curl}$-context, we need to deal with local gradient fields (not local constants) in every $\Omega^{(i)}$ and their consistency to a global gradient field (not a global constant) instead. The underlying idea stays the same and we quantify that by relation \eqref{eq:gloLocKernel}.

Additional issues can arise due to edges which are shared by more than two subdomains. In analogy to the ``cross-points'' in nodal problems, we call them ``cross-edges'' and denote them as $\crossEdges$ in the global and as $\crossEdges^{(i)}$ in the local context. These introduce redundancy in the coupling, i.e. it is not possible to uniquely determine $\mults$ from \eqref{eq:TIproblem} because $\operatorname{ker}\left(\couplMat\trans\right)\neq\{\mathbf{0}\}$. As stated in \cite[Sec.~6.3.1]{Toselli_2005aa} or \cite[Thm.~3.2]{Benzi_2005aa} and the appended remark therein, this is not a problem if one introduces the restriction $\mults\in\operatorname{im}\left(\couplMat\right)$ and if $\ker\left(\stiffMat\right)\cap\ker\left(\couplMat\right)=\{\mathbf{0}\}$ holds. Some configurations of cross-edges do not satisfy the latter condition while other remain solvable, but in most literature the corresponding coupling constraints are modified in some way to obtain appropriate solutions \cite{Toselli_2005aa}. One elegant approach is the dual-primal TI variant in which selected constraints are eliminated by reducing each set of coupled DOFs in the selection to one DOF, \cite{Klawonn_2002aa,Farhat_2001aa}. In the nodal context, this approach has the additional advantage that the unspecified local constants in floating subdomains are fixed if the primal DOFs are selected appropriately.

This last additional advantage of the dual-primal ansatz for nodal problems motivates the following procedure in the $\mathrm{curl}$-$\mathrm{curl}$-context. First, construct an appropriate tree on $\graph{}$ that can be used to gauge the global problem. Then, tear the global tree into its local components on $\graph{}^{(i)}$ and show that combining these with the dual-primal idea enables local gauging, i.e. invertibility of the local operator for each subdomain. A more detailed description of tree-cotree gauging and our choice for the global tree is given in the following part of this section. Considerations and details on the dual-primal approach with tree-cotree gauging are provided in \autoref{sec:dualPrimal}.}

\subsection{Global Tree-Cotree Decomposition for Dual-Primal Approach}
\revised{For the following considerations, we need to classify the edges and nodes of $\graph{\Gamma}$ according to which boundary parts their respective basis functions contribute to. Starting in the local context, we write $\nodeSet{D}^{(i)}$ and $\edgeSet{D}^{(i)}$ for nodes and edges \ma{whose corresponding basis functions do not vanish on the Dirichlet boundary}. Similarly, $\nodeSet{N}^{(i)}$, $\edgeSet{N}^{(i)}$ for the Neumann parts and $\nodeSet{I}^{(i)}$, $\edgeSet{I}^{(i)}$ for nodes and edges with relation to interfaces. Note that their intersections are not necessarily empty. The global counterparts $\nodeSet{\star}$ and $\edgeSet{\star}$ for $\star\in\{\mathrm{D},\mathrm{I},\mathrm{N}\}$ are defined via the union of the respective local ones.

In tree-cotree gauging, one exploits the fact that it is possible to identify a decomposition of the DOFs into determined and undetermined ones from topological information of the mesh, i.e. $\graph{}$, alone. Edges belonging to a spanning tree on $\graph{}$ correspond to undetermined DOFs while the remaining (determined) ones are called cotree DOFs, \cite{Manges_1995aa,Albanese_1988ab}. We write $\edgeSet{t}\subset\edgeSet{}$ for the tree edges and $\edgeSet{c}=\edgeSet{}\Setminus\edgeSet{t}$ for the cotree ones in the global context. The gauge then consists of prescribing the tree DOFs and eliminating them from the system. We employ Kruskal's algorithm \cite{Kruskal_1956aa,Cormen_2001aa} with appropriate edge weights to grow such a spanning tree. The idea behind this algorithm works as follows: A not yet visited edge with the lowest weight is added to the tree if it does not close a loop. This is repeated until a stopping criterion, e.g. $\cardSet{\edgeSet{t}}=\cardSet{\nodeSet{}}-1$ for a connected graph $\graph{}$, is reached. Following \cite[Sec.~24.2]{Cormen_2001aa}, one can expect $\mathcal{O}\left(n_{\mathrm{g}}\log\left(n_{\mathrm{g}}\right)\right)$ with $n_{\mathrm{g}}=\cardSet{\edgeSet{}}$ for the duration of Kruskal's algorithm in an efficient implementation.

The elimination of Dirichlet boundary conditions conflicts with the elimination of tree DOFs because more DOFs than necessary are prescribed with a value. This leads to restrictions for the choice of the spanning tree. E.g. in \cite{Dular_1995aa}, it is described that the tree subset $\edgeSet{Dt}\coloneqq\edgeSet{D}\cap\edgeSet{t}$ needs to form a spanning tree on the subgraph $\graph{D}\coloneqq\left(\nodeSet{D},\edgeSet{D}\right)$ to consistently gauge \eqref{eq:discVecPot}. Such a tree can be obtained by constructing it on $\Gamma_{\mathrm{D}}$ first, before extending it into the remaining parts of $\Omega$, \cite{Dular_1995aa}. Kruskal's algorithm preserves this hierarchy by assigning the lowest weight to edges with relation to the Dirichlet boundary and assigning higher values to the remaining edges.

Taking care of the Dirichlet contribution correctly enables an appropriate gauge of \eqref{eq:discVecPot}, but increasing the number of hierarchy levels, i.e. weights in Kruskal's algorithm, provides good properties which we can exploit in the local context. We propose to grow the tree on the wire basket graph $\graph{\lambda}$ first, before extending it into the facets, represented by $\graph{\Gamma}$, and after that into the volume, i.e. covering the full $\graph{}$. This reduces complexity and allows for parallelism since the treatment of facets and volumes is automatically independent. More importantly, we can state and prove \autoref{thm:constCotree} for the cotree edges on the wire basket $\edgeSet{\lambda c}\coloneqq\edgeSet{\lambda}\Setminus\edgeSet{t}$.
\begin{theorem}\label{thm:constCotree}
    Let $\edgeSet{t}$ be a set of edges which form a spanning tree on the subgraph $\graph{\lambda}$ of the wire basket. Then, $\cardSet{\edgeSet{\lambda c}}$ is independent of the mesh size $h$ and degree $p$, i.e. of $\cardSet{\edgeSet{}}$, for a given decomposition of $\Omega$.
\end{theorem}
\begin{proof}
    We can consider the decomposition of $\Omega$ into patches as a coarse mesh, or more precisely a CW-complex, where each patch corresponds to one hexahedral cell. We respectively denote by $\nodeSet{dec}$, $\edgeSet{dec}$, $\faceSet{dec}$ and $\cellSet{dec}$ the sets of nodes, edges, faces and cells for this mesh. The corresponding graph is denoted as $\graph{dec}\coloneqq\left(\nodeSet{dec},\edgeSet{dec}\right)$. Using Euler's characteristic formula \cite[Thm.~2.44]{Hatcher_2001aa}, we have
    \begin{equation}
        \cardSet{\nodeSet{dec}} - \cardSet{\edgeSet{dec}} + \cardSet{\faceSet{dec}} - \cardSet{\cellSet{dec}} = c_1 \in \mathbb{Z},\label{eq:eulerID_1}
    \end{equation}
    where the constant $c_1$ depends on the topology of the domain, but is independent of the decomposition into subdomains. Note that we can obtain $\graph{\lambda}$ by subdividing every edge of $\graph{dec}$ appropriately. Consequently, we can interpret $\graph{\lambda}$ as the edges and vertices of a CW-complex of the same domain, with the same faces $\faceSet{dec}$ and cells $\cellSet{dec}$, and the same topological invariant $c_1$. Furthermore, we know that $\graph{\lambda}$ is a connected graph, implying $\cardSet{\edgeSet{\lambda t}}=\cardSet{\nodeSet{\lambda}}-1$ for a spanning tree on the wire basket. Consequently, we obtain
    \begin{equation}
        \cardSet{\edgeSet{\lambda c}}=\cardSet{\edgeSet{\lambda}}-\cardSet{\edgeSet{\lambda t}}=\cardSet{\edgeSet{\lambda}}-\cardSet{\nodeSet{\lambda}}+1
        = \cardSet{\faceSet{dec}} - \cardSet{\cellSet{dec}} - c_1 + 1,\label{eq:cotreeWireframe}
    \end{equation}
    due to \eqref{eq:eulerID_1}. This proves \autoref{thm:constCotree} because $\cardSet{\edgeSet{\lambda c}}$ only depends on the topology and the given decomposition of $\Omega$ but not on $\cardSet{\edgeSet{}}$, i.e. mesh size $h$ and degree $p$.
\end{proof}
\autoref{thm:constCotree} can be applied to more general settings than the IGA multipatch decomposition of $\Omega$ but we neglect a detailed discussion. As a consequence of \autoref{thm:constCotree} and its proof, we can state \autoref{cor:linearCompl} as an extension.
\begin{corollary}\label{cor:linearCompl}
    If the decomposition of $\Omega$ satisfies $\cardSet{\faceSet{dec}}=\mathcal{O}(N)$, then $\cardSet{\edgeSet{\lambda c}}=\mathcal{O}(N)$ for the number of subdomains $N$.
\end{corollary}
\begin{proof}
    Clearly, we have that $\cardSet{\cellSet{dec}}=N$ and that $c_1$ is independent of $N$. Hence, $\cardSet{\edgeSet{\lambda c}}=\mathcal{O}(N)$ follows from \eqref{eq:cotreeWireframe} if $\cardSet{\faceSet{dec}}=\mathcal{O}(N)$.
\end{proof}
In the IGA multipatch setting, we know that every patch, i.e. every subdomain, has six faces. The sum of all patches counts every interface twice such that $0<\cardSet{\faceSet{dec}}=6N-n_{\mathrm{intf}}=\mathcal{O}(N)$ because $n_{\mathrm{intf}}<6N$ holds for the number of interfaces $n_{\mathrm{intf}}$. Therefore, we know that
\begin{equation}
    \cardSet{\edgeSet{\lambda c}}=\mathcal{O}(N)\label{eq:linCompl}
\end{equation}
holds for the IGA multipatch decomposition of $\Omega$ due to \autoref{cor:linearCompl}.

Let us address the question whether a spanning tree $\edgeSet{t}$ which provides also a spanning tree $\edgeSet{\lambda t}$ when restricted to the wire basket $\graph{\lambda}$ can be constructed to also provide a spanning tree $\edgeSet{Dt}$ when restricted to the Dirichlet boundary $\graph{D}$. To answer this it is necessary to introduce a classification of different wire basket subsets depending on which facet type they contribute to. First, we denote the intersections of different facet types as $\edgeSet{DI}\coloneqq\edgeSet{D}\cap\edgeSet{I}$, $\edgeSet{DN}\coloneqq\edgeSet{D}\cap\edgeSet{N}$ and $\edgeSet{NI}\coloneqq\edgeSet{N}\cap\edgeSet{I}$ respectively. Then, if edges are shared by two facets of the same type, we write $\edgeSet{DD}\coloneqq\left(\edgeSet{\lambda}\cap\edgeSet{D}\right)\Setminus\left(\edgeSet{DI}\cup\edgeSet{DN}\right)$, $\edgeSet{NN}\coloneqq\left(\edgeSet{\lambda}\cap\edgeSet{N}\right)\Setminus\left(\edgeSet{DN}\cup\edgeSet{NI}\right)$ and $\edgeSet{II}\coloneqq\left(\edgeSet{\lambda}\cap\edgeSet{I}\right)\Setminus\left(\edgeSet{DI}\cup\edgeSet{NI}\right)$, respectively. Altogether, the relations
\begin{equation*}
    \edgeSet{\lambda}=\edgeSet{DI}\cup\edgeSet{DN}\cup\edgeSet{NI}\cup\edgeSet{II}\cup\edgeSet{DD}\cup\edgeSet{NN},\quad\edgeSet{\Gamma}=\edgeSet{D}\cup\edgeSet{N}\cup\edgeSet{I}\quad\text{and}\quad\edgeSet{\lambda}\subseteq\edgeSet{\Gamma}
\end{equation*}
hold. Note that the corresponding node subsets, e.g. $\nodeSet{NN}$, are considered to be defined similarly. Through introducing a hierarchy on growing the tree first on $\edgeSet{DI}$, then $\edgeSet{DN}$ and at last to $\edgeSet{NI}$, we preserve the ``Dirichlet first'' condition on the wire basket and the remaining tree can be grown into the facets independently. After that, we extend the tree to $\edgeSet{DD}$, $\edgeSet{NN}$ and $\edgeSet{II}$ to obtain that $\edgeSet{\lambda t}$ is a \revised{spanning} tree on $\graph{\lambda}$ such that \autoref{thm:constCotree} can be employed. Note that $\edgeSet{II}$, $\edgeSet{DD}$ and $\edgeSet{NN}$ can be handled independently of each other and that we assign different weights only for convenience. The remaining steps consist of extending the tree to $\graph{\Gamma}$ and then to $\graph{}$. This hierarchy can be obtained by employing Kruskal's algorithm with the presented weights in \autoref{fig:kruskalWeights}. Note that it is possible to parallelize the extension into certain subgraphs such that the general complexity $\mathcal{O}\left(n_{\mathrm{g}}\log\left(n_{\mathrm{g}}\right)\right)$ can be improved. We have visualized the weights and the resulting tree for an example configuration in \autoref{fig:mixedCube_loc}.}

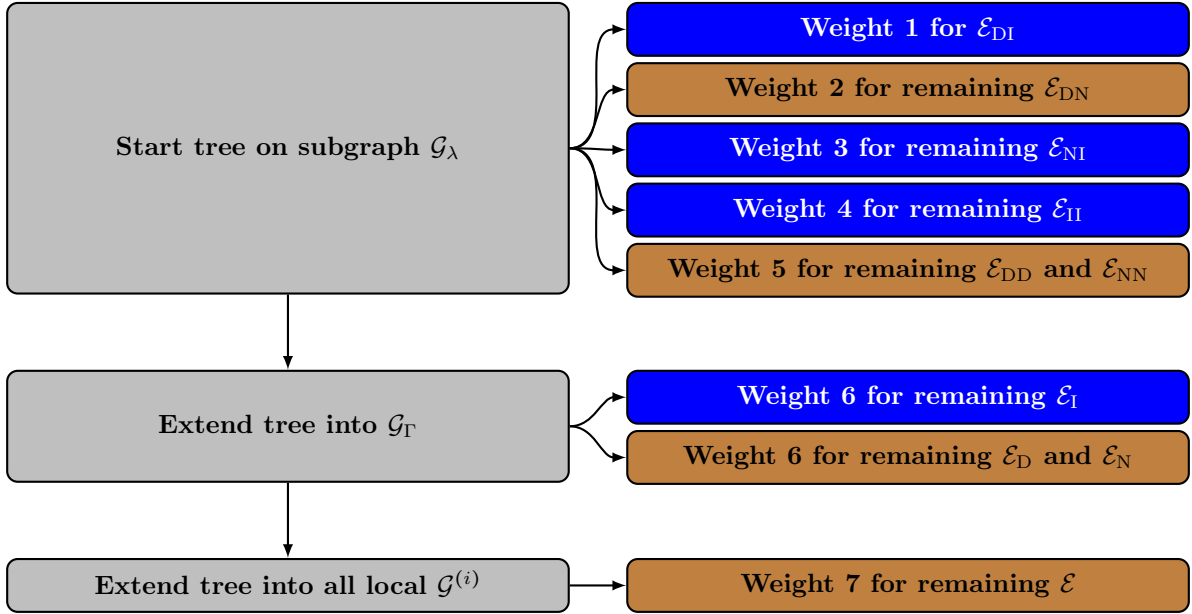
\begin{figure}
    \centering
    \begin{tikzpicture}[>=latex]
        \tikzset{textBox/.style={rectangle, rounded corners, draw, thick, minimum width=0.45\linewidth, inner sep=0pt, font={\bf}}}
        
        \node[textBox,fill=lightgray,minimum height=11em](n1) at (0,0){Start tree on subgraph $\graph{\lambda}$};
        \node[textBox,fill=lightgray,minimum height=4.2em,  anchor=north](n2) at ($(n1.south)-(0,1)$){Extend tree into $\graph{\Gamma}$};
        \node[textBox,fill=lightgray,minimum height=2em,anchor=north](n3) at ($(n2.south)-(0,1)$){Extend tree into all local $\graph{}^{(i)}$};

        \node[textBox,fill=blue,text=white,minimum height=2em,anchor=north west](n4) at ($(n1.north east)+(0.75,0)$){Weight 1 for $\mathcal{E}_{\mathrm{DI}}$};
        \node[textBox,fill=brown,text=black,minimum height=2em,anchor=north](n5) at ($(n4.south)-(0em,0.2em)$){Weight 2 for remaining $\mathcal{E}_{\mathrm{DN}}$};
        \node[textBox,fill=blue,text=white,minimum height=2em,anchor=north](n6) at ($(n5.south)-(0em,0.2em)$){Weight 3 for remaining $\mathcal{E}_{\mathrm{NI}}$};
        \node[textBox,fill=blue,text=white,minimum height=2em,anchor=north](n7) at ($(n6.south)-(0em,0.2em)$){Weight 4 for remaining $\edgeSet{II}$};
        \node[textBox,fill=brown,minimum height=2em,anchor=north](n8) at ($(n7.south)-(0em,0.2em)$){Weight 5 for remaining $\edgeSet{DD}$ and $\edgeSet{NN}$};

        \node[textBox,fill=blue,text=white,minimum height=2em,anchor=north west](n9) at ($(n2.north east)+(0.75,0)$){Weight 6 for remaining $\edgeSet{I}$};
        \node[textBox,fill=brown,minimum height=2em,anchor=north](n10) at ($(n9.south)-(0em,0.2em)$){Weight 6 for remaining $\edgeSet{D}$ and $\edgeSet{N}$};

        \node[textBox,fill=brown,minimum height=2em,anchor=west](n11) at ($(n3.east)+(0.75,0)$){Weight 7 for remaining $\edgeSet{}$};

        \draw[->,thick] (n1.south) -- (n2.north);
        \draw[->,thick] (n2.south) -- (n3.north);
        
        \draw[->,thick] (n1.east) to[in=180,out=0] (n4.west);
        \draw[->,thick] (n1.east) to[in=180,out=0] (n5.west);
        \draw[->,thick] (n1.east) to[in=180,out=0] (n6.west);
        \draw[->,thick] (n1.east) to[in=180,out=0] (n7.west);
        \draw[->,thick] (n1.east) to[in=180,out=0] (n8.west);

        \draw[->,thick] (n2.east) to[in=180,out=0] (n9.west);
        \draw[->,thick] (n2.east) to[in=180,out=0] (n10.west);

        \draw[->,thick] (n3.east) to[in=180,out=0] (n11.west);
    \end{tikzpicture}
    \caption{Weight distribution for Kruskal's algorithm to generate a consistent tree. Tree generation for brown boxes can be parallelized for every subdomain. The blue boxes need global context and no parallelization is possible.}
    \label{fig:kruskalWeights}
\end{figure}

\begin{figure}
    \centering
    \includegraphics[keepaspectratio,width=0.95\linewidth]{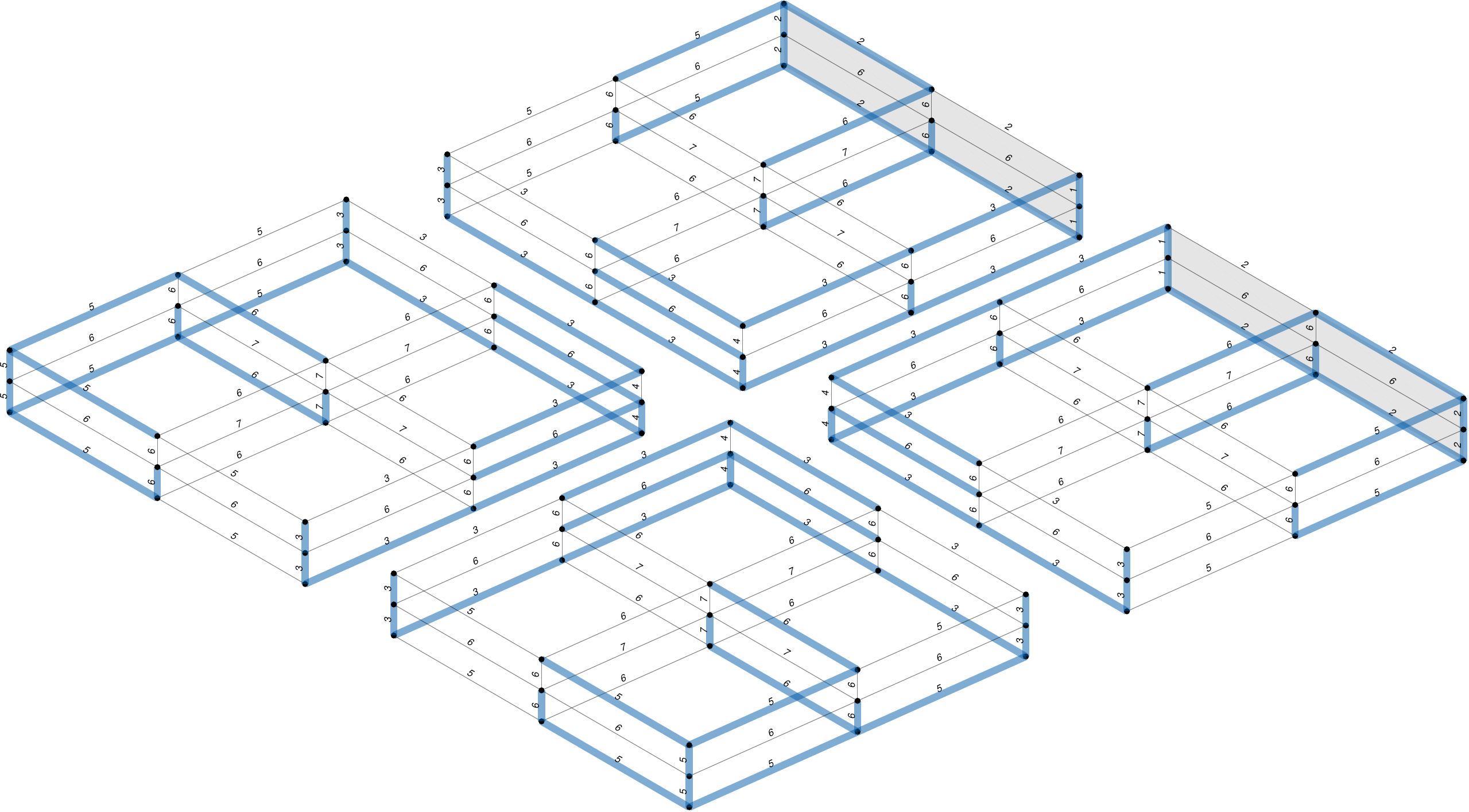}
    \caption{Local graphs with prioritization, tree edges (blue) and Dirichlet sides marked in gray.}
    \label{fig:mixedCube_loc}
\end{figure}

%% file: dpForm.tex
\revised{In this section, we derive and explain the typical dual-primal approach in \autoref{subsec:dpTI}. This method requires selecting appropriate, constrained DOFs for which continuity is enforced by reducing them to only one DOF. These are denoted as `primal' and we present such a choice in \autoref{subsec:locInv}. Additionally, we prove that our choice satisfies all requirements to obtain a scalable TI method.}

\subsection{Dual-Primal Tearing and Interconnecting}\label{subsec:dpTI}
Starting with \eqref{eq:TIproblem}, where the Dirichlet boundary is already eliminated, we further split $\dofs$ into tree DOFs $\dofs_{\mathrm{t}}$, primal DOFs $\dofs_{\mathrm{p}}$ and the remaining ones $\dofs_{\mathrm{r}}$. \revised{We write $n_{\mathrm{p}}$ for the number of the primal DOFs which depends on the, for now unspecified, selection. Incorporating $\dofs_{\mathrm{t}}=\mathbf{0}$ and reordering the equations leads to
\begin{equation}
    \begin{bmatrix}
        \stiffMat_{\mathrm{rr}} & \stiffMat_{\mathrm{rp}} & \bar{\couplMat}_{\mathrm{r}}\trans  \\
        \stiffMat_{\mathrm{pr}} & \stiffMat_{\mathrm{pp}} & \bar{\couplMat}_{\mathrm{p}}\trans \\
        \bar{\couplMat}_{\mathrm{r}} & \bar{\couplMat}_{\mathrm{p}} & \mathbf{0} \\
    \end{bmatrix}\begin{bmatrix}
        \dofs_{\mathrm{r}} \\
        \dofs_{\mathrm{p}} \\
        \bar{\mults} \\
    \end{bmatrix}=\begin{bmatrix}
        \rhs_{\mathrm{r}} \\
        \rhs_{\mathrm{p}} \\
        \mathbf{0} \\
    \end{bmatrix}.\label{eq:intermedDP}
\end{equation}
Note that removing tree-DOFs is similar to eliminating Dirichlet DOFs in the sense that zero rows are removed again to obtain the reduced coupling matrix $\bar{\couplMat}=\begin{bmatrix}\bar{\couplMat}_{\mathrm{r}} & \bar{\couplMat}_{\mathrm{p}}\end{bmatrix}$. Hence, we write $\bar{\couplMat}_{\mathrm{p}}$ and $\bar{\couplMat}_{\mathrm{r}}$ for the respective primal and remaining modified coupling constraints as well as $\bar{\mults}$ for the multipliers. We define a matrix $\primMat\in\{0,1\}^{n_{\mathrm{p}}\times \nCoarse}$ which is the kernel representation of $\bar{\couplMat}_{\mathrm{p}}$, i.e. $\operatorname{ker}\left(\bar{\couplMat}_{\mathrm{p}}\right)=\operatorname{im}\left(\primMat\right)$ with $\nCoarse=\dimSpace{\operatorname{ker}\left(\bar{\couplMat}_{\mathrm{p}}\right)}$. Substituting $\dofs_{\mathrm{p}}=\primMat\mathbf{p}$ in \eqref{eq:intermedDP}, removing the constraints associated with the zero rows in $\bar{\couplMat}_{\mathrm{r}}$ and multiplying the second line of \eqref{eq:intermedDP} with $\primMat$ yields}
\begin{equation}
    \begin{bmatrix}
        \stiffMat_{\mathrm{rr}} & \stiffMat_{\mathrm{rp}}\primMat & \couplMat_{\mathrm{r}}\trans  \\
        \primMat\trans\stiffMat_{\mathrm{pr}} & \primMat\trans\stiffMat_{\mathrm{pp}}\primMat & \mathbf{0} \\
        \couplMat_{\mathrm{r}} & \mathbf{0} & \mathbf{0} \\
    \end{bmatrix}\begin{bmatrix}
        \dofs_{\mathrm{r}} \\
        \mathbf{p} \\
        \mults_{\mathrm{r}} \\
    \end{bmatrix}=\begin{bmatrix}
        \rhs_{\mathrm{r}} \\
        \primMat\trans\rhs_{\mathrm{p}} \\
        \mathbf{0} \\
    \end{bmatrix}\label{eq:gloDP}
\end{equation}
as the global dual-primal system with remaining multipliers $\mults_{\mathrm{r}}\in\mathbb{R}^{m_{\mathrm{r}}}$. \revised{These steps effectively reduce each set of coupled DOFs in $\dofs_{\mathrm{p}}$ to one DOF in $\mathbf{p}$.}

To enable parallelism by Schur complements in a typical dual-primal TI method, one may require that 
\begin{equation*}
    \stiffMat_{\mathrm{rr}}=\operatorname{diag}\left(\stiffMat_{\mathrm{rr}}^{(1)},~\ldots~,~\stiffMat_{\mathrm{rr}}^{(N)}\right),
\end{equation*}
\ma{where every} local matrix $\stiffMat_{\mathrm{rr}}^{(i)}$ \revised{is invertible. In other words, fixing the tree DOFs and selecting appropriate primal DOFs needs to remove all local kernels.} Assuming that this holds, one can eliminate
\begin{equation}
    \dofs_{\mathrm{r}} = \stiffMat_{\mathrm{rr}}^{-1}\left(\rhs_{\mathrm{r}} - \stiffMat_{\mathrm{rp}}\primMat\mathbf{p} - \couplMat_{\mathrm{r}}\trans\mults_{\mathrm{r}} \right)\label{eq:recRem}
\end{equation}
in the first line of \eqref{eq:gloDP}. Inserting \eqref{eq:recRem} in the second and third line of \eqref{eq:gloDP} results in
\begin{equation}
    \begin{bmatrix}
        \primMat\trans\stiffMat_{\mathrm{pr}}\stiffMat_{\mathrm{rr}}^{-1}\stiffMat_{\mathrm{rp}}\primMat-\primMat\trans\stiffMat_{\mathrm{pp}}\primMat & \primMat\trans\stiffMat_{\mathrm{pr}}\stiffMat_{\mathrm{rr}}^{-1}\couplMat_{\mathrm{r}}\trans\\
        \couplMat_{\mathrm{r}}\stiffMat_{\mathrm{rr}}^{-1}\stiffMat_{\mathrm{rp}}\primMat & \couplMat_{\mathrm{r}}\stiffMat_{\mathrm{rr}}^{-1}\couplMat_{\mathrm{r}}\trans \\
    \end{bmatrix}\begin{bmatrix}
        \mathbf{p} \\
	    \mults_{\mathrm{r}} \\
    \end{bmatrix} = \begin{bmatrix}
        \primMat\trans\stiffMat_{\mathrm{pr}}\stiffMat_{\mathrm{rr}}^{-1}\rhs_{\mathrm{r}} - \primMat\trans\rhs_{\mathrm{p}} \\
        \couplMat_{\mathrm{r}}\stiffMat_{\mathrm{rr}}^{-1}\rhs_{\mathrm{r}}
    \end{bmatrix}.
\end{equation}
We shorten the notation by introducing 
\begin{align*}
    \coarseMat&\coloneqq\primMat\trans\stiffMat_{\mathrm{pr}}\stiffMat_{\mathrm{rr}}^{-1}\stiffMat_{\mathrm{rp}}\primMat-\primMat\trans\stiffMat_{\mathrm{pp}}\primMat, \\
    \mathbf{G}&\coloneqq\couplMat_{\mathrm{r}}\stiffMat_{\mathrm{rr}}^{-1}\stiffMat_{\mathrm{rp}}\primMat, \\
    \mathbf{W}&\coloneqq\couplMat_{\mathrm{r}}\stiffMat_{\mathrm{rr}}^{-1}\couplMat_{\mathrm{r}}\trans, \\
    \mathbf{d}&\coloneqq\primMat\trans\stiffMat_{\mathrm{pr}}\stiffMat_{\mathrm{rr}}^{-1}\rhs_{\mathrm{r}} - \primMat\trans\rhs_{\mathrm{p}},\\
    \mathbf{e}&\coloneqq\couplMat_{\mathrm{r}}\stiffMat_{\mathrm{rr}}^{-1}\rhs_{\mathrm{r}},
\end{align*}
which leads to
\begin{equation}
    \begin{bmatrix}
        \coarseMat & \mathbf{G}\trans \\
        \mathbf{G} & \mathbf{W} \\
    \end{bmatrix}\begin{bmatrix}
        \mathbf{p} \\
	    \mults_{\mathrm{r}} \\
    \end{bmatrix} = \begin{bmatrix}
        \mathbf{d} \\
        \mathbf{e} \\
    \end{bmatrix}.\label{eq:firstSchur}
\end{equation}
For this, we apply yet again a Schur complement to eliminate $\mathbf{p}$ and obtain
\begin{align}
    \interfaceMat\mults_{\mathrm{r}} = \left(\mathbf{G}\coarseMat^{-1}\mathbf{G}\trans-\mathbf{W}\right)\mults_{\mathrm{r}} &= \mathbf{G}\coarseMat^{-1}\mathbf{d} - \mathbf{e}\label{eq:intProb} \\
    \mathbf{a}_{\mathrm{p}} &= \primMat\coarseMat^{-1}\left(\mathbf{d} - \mathbf{G}\trans\mults_{\mathrm{r}}\right)\label{eq:coarseProb} \\
    \dofs_{\mathrm{r}} &= \stiffMat_{\mathrm{rr}}^{-1}\left(\rhs_{\mathrm{r}} - \stiffMat_{\mathrm{rp}}\mathbf{a}_{\mathrm{p}} - \couplMat_{\mathrm{r}}\trans\mults_{\mathrm{r}} \right).\label{eq:localProb}
\end{align}
Note that the inherent parallelism becomes evident when all matrices are split into local contributions. Details are omitted but an interested reader is referred to \cite{Farhat_2001aa}.

\revised{Typically, \eqref{eq:intProb} is solved using iterative methods such as the preconditioned conjugate gradient (PCG) method as described in, e.g. \cite{Farhat_2000aa,Farhat_2001aa}. The condition number of the interface problem has a major influence on the scalability of the method which motivates looking at appropriate preconditioners in more detail. We follow the approach of \citeauthor{Farhat_2000aa} in \cite{Farhat_2000aa} and investigate the lumped
\begin{equation}
    \mathbf{M}_{\mathrm{L}}^{-1}\coloneqq\mathbf{D}_{\mathrm{s}}\couplMat_{\mathrm{r_{I}}}\stiffMat_{\mathrm{r_{I}}\mathrm{r_{I}}}\couplMat_{\mathrm{r_{I}}}\trans\mathbf{D}_{\mathrm{s}}
\end{equation}
as well as the Dirichlet preconditioner
\begin{equation}
    \mathbf{M}_{\mathrm{D}}^{-1}\coloneqq\mathbf{D}_{\mathrm{s}}\couplMat_{\mathrm{r_{I}}}\interfaceMat_{\mathrm{r_{I}}\mathrm{r_{I}}}\couplMat_{\mathrm{r_{I}}}\trans\mathbf{D}_{\mathrm{s}}\quad\text{with}\quad\interfaceMat_{\mathrm{r_{I}}\mathrm{r_{I}}}\coloneqq\stiffMat_{\mathrm{r_{I}}\mathrm{r_{I}}} - \stiffMat_{\mathrm{r_{I}}\mathrm{r_{V}}}\stiffMat_{\mathrm{r_{V}}\mathrm{r_{V}}}^{-1}\stiffMat_{\mathrm{r_{V}}\mathrm{r_{I}}}
\end{equation}
in \autoref{subsec:scalab}. For both approaches, it is necessary to further split the remaining DOFs $\dofs_{\mathrm{r}}$ into $\dofs_{\mathrm{r_{I}}}$, those which belong to an interface, and $\dofs_{\mathrm{r_{V}}}$, those which are not constrained in any way. The diagonal scaling matrix $\mathbf{D}_{\mathrm{s}}$ is necessary to treat material jumps between subdomains, \cite{Farhat_2000aa,Klawonn_2002aa}. In our numerical examples, we focus on problems with homogeneous material and use the identity matrix $\mathbf{D}_{\mathrm{s}}=\mathbf{I}$. More general scaling matrices can be constructed by using ``weighted counting functions'' as explained in \cite{Klawonn_2002aa}.}

\subsection{Selection of Primal DOFs and Invertibility}\label{subsec:locInv}
In this subsection, we discuss the choice of appropriate primal DOFs. \revised{We know that the global problem \eqref{eq:discVecPot} is gauged appropriately because our tree is constructed to provide that. One can interpret \eqref{eq:discVecPot} as a special variant of \eqref{eq:gloDP} where all interface DOFs are selected as primal. In other words, after eliminating a globally consistent tree as constructed in \autoref{sec:treeGraph}, the system in \eqref{eq:gloDP} is solvable because no issues with cross-edges remain and global gradients are taken care of.} This implies that global invertibility of the system matrix in \eqref{eq:gloDP} only depends on our choice of primal DOFs $\dofs_{\mathrm{p}}$. \revised{The local contribution $\dofs_{\mathrm{p}}^{(i)}$ of subdomain $\Omega^{(i)}$ corresponds to the set of local edges $\edgeSet{p}^{(i)}$. Its global counterpart (computed with union operation) is called $\edgeSet{p}$ and its edges correspond to the DOFs in $\mathbf{p}$. We have that $n_{\mathrm{p}}=\sum_{i=1}^{N}\cardSet{\edgeSet{p}^{(i)}}$ and $\nCoarse=\cardSet{\edgeSet{p}}$ hold for the respective number of DOFs/edges.}

\revised{Before defining $\edgeSet{p}^{(i)}$ and $\edgeSet{p}$ properly, we analyze which remaining issues our selection needs to fix. First, to obtain global invertibility of \eqref{eq:gloDP}, it is necessary to remove issues related to the cross-edges $\edgeSet{II}^{(i)}$. Due to setting all local tree DOFs to zero and the matching of tree edges on interfaces (originating from tearing the global tree apart), all issues from the cross-edges $\edgeSet{II}^{(i)}\cap\edgeSet{t}^{(i)}$ are already fixed. The remaining ones $\edgeSet{II}^{(i)}\Setminus\edgeSet{t}^{(i)}$ need to be selected as primal edges to obtain global invertibility. However, this does not suffice to obtain local invertibility, i.e. that every local matrix $\stiffMat_{\mathrm{rr}}^{(i)}$ is invertible. An example is provided in \autoref{fig:flatDD} where the matrices corresponding to subdomains in the middle and the lower left are not invertible. It can be seen easily because the underlying graph $\graph{e}^{(i)}\coloneqq\left(\nodeSet{}^{(i)},\edgeSet{e}^{(i)}\right)$, formed by all the nodes from $\graph{}^{(i)}$ and the locally eliminated edges
\begin{equation}
    \edgeSet{e}^{(i)}\coloneqq\edgeSet{D}^{(i)}\cup\edgeSet{t}^{(i)}\cup\edgeSet{p}^{(i)},
\end{equation}
is not connected. In other words, there is no path from every node in $\nodeSet{}^{(i)}$ to all other nodes in $\nodeSet{}^{(i)}$ which only consists of edges in $\edgeSet{e}^{(i)}$ for the affected subdomains. We formalize this in \autoref{lem:locInv}.}

\begin{figure}
    \centering
    \includegraphics[keepaspectratio,width=0.95\linewidth]{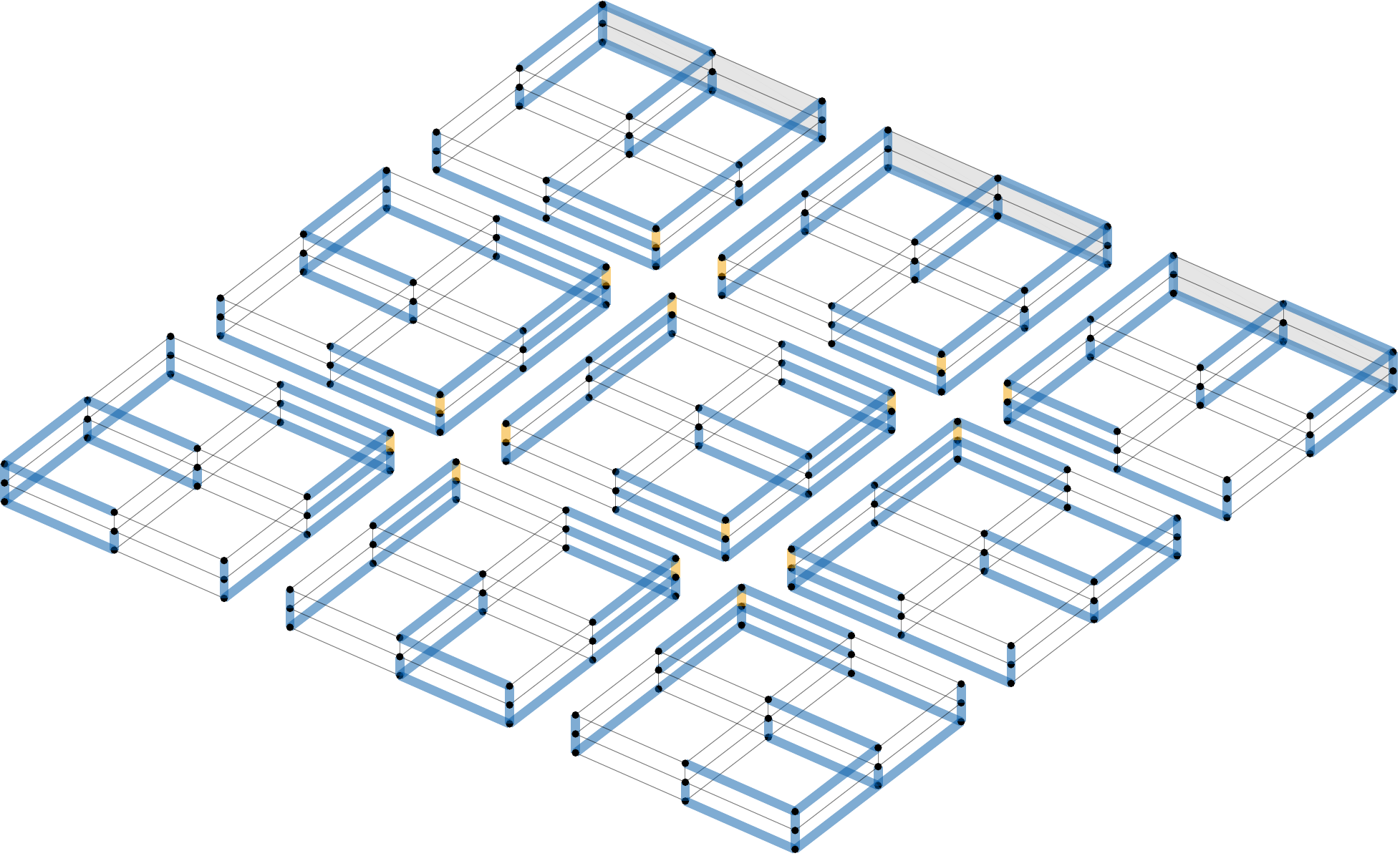}
    \caption{Local graphs for a problematic configuration with tree edges (blue), cross-edges (orange) and Dirichlet facets marked in gray.}
    \label{fig:flatDD}
\end{figure}
\revised{
\begin{lemma}\label{lem:locInv}
    The matrix $\stiffMat_{\mathrm{rr}}^{(i)}$ is invertible if the graph of locally eliminated DOFs $\graph{e}^{(i)}=\left(\nodeSet{}^{(i)},\edgeSet{e}^{(i)}\right)$ is connected.
\end{lemma}
\begin{proof}
    The proof is based on the arguments from \cite{Albanese_1988aa,Manges_1995aa}. If all the DOFs from a spanning tree are removed from a stiffness matrix belonging to the $\mathrm{curl}$-$\mathrm{curl}$-operator, it is invertible which implies that all its column vectors are linearly independent. A spanning tree has the fundamental property that it connects all nodes with a minimal amount of edges. If we eliminate the DOFs related to a spanning tree along with some additional edges, we know that the resulting matrix is still invertible because removing one column and one row of the stiffness matrix preserves the linear independence of its column vectors. If $\graph{e}^{(i)}$ is connected, and since it contains all the nodes in $\nodeSet{}^{(i)}$, we implicitly know that $\edgeSet{e}^{(i)}$ contains a spanning tree on $\graph{}^{(i)}$. Hence, the statement for invertibility can be applied to the local stiffness matrix $\stiffMat^{(i)}$ and its reduced form $\stiffMat_{\mathrm{rr}}^{(i)}$ with graph $\graph{e}^{(i)}$, which concludes the proof.
\end{proof}
\begin{remark}
    Eliminating more DOFs than necessary for regularizing the corresponding stiffness matrix gives rise to issues regarding consistency. In other words, only a certain amount of DOFs can be chosen freely when solving a singular system. The ones which are additionally eliminated need to be prescribed consistently. In our context, these additional DOFs are either primal or belong to the Dirichlet boundary. The Dirichlet DOFs are automatically consistent if the boundary condition is consistent with the original problem. The primal DOFs are computed by solving the coarse problem \eqref{eq:coarseProb} which preserves the consistent, global context.
\end{remark}

Now, all tools are ready to define our selection of primal edges $\edgeSet{p}^{(i)}$, which we do in \autoref{thm:primEdges}: we prove that global and local invertibility can be obtained with our algorithm.

\begin{theorem}\label{thm:primEdges}
    Let the tree be generated with Kruskal's algorithm and weights as in \autoref{fig:kruskalWeights}, and let the local primal edges be chosen as 
    \begin{equation}
        \edgeSet{p}^{(i)}=\left(\crossEdges^{(i)}\cup\edgeSet{NI}^{(i)}\right)\Setminus\edgeSet{t}^{(i)}.
    \end{equation}
    Then system \eqref{eq:gloDP} satisfies global and local invertibility.
\end{theorem}
\begin{proof}
    Global invertibility originates from the fact that all edges $\crossEdges^{(i)}\Setminus\edgeSet{t}^{(i)}$ are contained in $\edgeSet{p}^{(i)}$ and that the tree edges $\edgeSet{t}$ and correspondingly $\edgeSet{t}^{(i)}$ are chosen such that the multipatch problem is gauged appropriately.

    For local invertibility we use \autoref{lem:locInv}, for which we need to show that the combination of the tree construction and our final selection of primal DOFs yields connected graphs $\graph{e}^{(i)}$ on every subdomain. The proof is carried out separately for every subdomain $\Omega^{(i)}$ in three steps. The three steps are respectively: First, showing that all nodes $\nodeSet{\lambda}^{(i)}$ are connected by edges in $\edgeSet{\lambda e}^{(i)}\coloneqq\edgeSet{\lambda}^{(i)}\cap\edgeSet{e}^{(i)}$. Second, that the extension of Kruskal's algorithm onto $\nodeSet{\Gamma}^{(i)}$ and the elimination of all Dirichlet DOFs connects all nodes $\nodeSet{\Gamma}^{(i)}$ with edges in $\edgeSet{\Gamma e}^{(i)}\coloneqq\edgeSet{\Gamma}^{(i)}\cap\edgeSet{e}^{(i)}$. Third, that the extension of the tree into the interior yields a connected graph $\graph{e}^{(i)}$ for every subdomain.

    By reformulating $\edgeSet{\lambda e}^{(i)}$, one obtains
    \begin{align*}
        \edgeSet{\lambda e}^{(i)}&=\left(\edgeSet{D}^{(i)}\cap\edgeSet{\lambda}^{(i)}\right)\cup\left(\edgeSet{t}^{(i)}\cap\edgeSet{\lambda}^{(i)}\right)\cup\left(\edgeSet{p}^{(i)}\cap\edgeSet{\lambda}^{(i)}\right) \\
        &=\left(\edgeSet{DD}^{(i)}\cup\edgeSet{DI}^{(i)}\cup\edgeSet{DN}^{(i)}\right)\cup\left(\left(\crossEdges^{(i)}\cup\edgeSet{NI}^{(i)}\cup\edgeSet{NN}^{(i)}\right)\cap\edgeSet{t}^{(i)}\right)\cup\left(\left(\crossEdges^{(i)}\cup\edgeSet{NI}^{(i)}\right)\Setminus\edgeSet{t}^{(i)}\right) \\
        &=\underbrace{\edgeSet{DD}^{(i)}\cup\edgeSet{DI}^{(i)}\cup\edgeSet{DN}^{(i)}\cup\edgeSet{NI}^{(i)}\cup\edgeSet{II}^{(i)}}_{=\edgeSet{\lambda}^{(i)}\Setminus\edgeSet{NN}^{(i)}}\cup\left(\edgeSet{NN}^{(i)}\cap\edgeSet{t}^{(i)}\right).
    \end{align*}
    This implies that all nodes in $\nodeSet{\lambda}^{(i)}\Setminus\nodeSet{NN}^{(i)}$ are already connected because all available edges $\edgeSet{\lambda}^{(i)}\Setminus\edgeSet{NN}^{(i)}$ are contained in $\edgeSet{\lambda e}^{(i)}$. The remaining ones $\nodeSet{NN}^{(i)}$ are also connected because Kruskal's algorithm with given weights constructs a tree on $\graph{\lambda}$ first. In detail, this is valid because all nodes in $\nodeSet{NN}^{(i)}$ and all edges in $\edgeSet{NN}^{(i)}$ stay the same in both the local as well as the global context (no components of the sets are removed in the union operation). In other words, all nodes in $\nodeSet{NN}$ are connected by edges in $\edgeSet{NN}\cap\edgeSet{t}$ which in turn implies the same for the respective local sets. Altogether, we therefore know that all nodes of $\nodeSet{\lambda}^{(i)}$ are connected by edges in $\edgeSet{\lambda e}^{(i)}$.

    For the second step, we need to show that all nodes in $\nodeSet{\Gamma}^{(i)}$ are connected by edges in $\edgeSet{\Gamma e}^{(i)}$. We already know that this holds for $\nodeSet{\lambda}^{(i)}\subseteq\nodeSet{\Gamma}^{(i)}$ with edges from $\edgeSet{\lambda e}^{(i)}\subseteq\edgeSet{\Gamma e}^{(i)}$.  Next, Kruskal's algorithm extends the tree into every facet independently which connects all nodes in $\nodeSet{\Gamma}^{(i)}$ to the already available tree (whose edges are part of the eliminated edges). Furthermore, Dirichlet edges are eliminated which does not influence the connectedness. Consequently, all nodes in $\nodeSet{\Gamma}^{(i)}$ are connected by edges in $\edgeSet{\Gamma e}^{(i)}$.

    The same arguments can be employed in the third step, where extending the tree into every subdomain implies that all nodes in $\nodeSet{}^{(i)}$ are connected by edges in $\edgeSet{e}^{(i)}$. Hence, all graphs $\graph{e}^{(i)}$ are connected and every local block in $\stiffMat_{\mathrm{rr}}$ is invertible due to \autoref{lem:locInv}.
\end{proof}
\begin{remark}\label{rem:effPrimal}
    Our choice of $\edgeSet{p}^{(i)}$ is both straightforward to implement and efficient. It is simple because we do not need an algorithm to search for edges that are necessary to obtain local connectedness. At the same time, we may select more edges than needed for local connectedness but we know $\edgeSet{p}\subseteq\edgeSet{\lambda c}$. Due to this and \eqref{eq:linCompl} we know that
    \begin{equation}
        \nCoarse=\cardSet{\edgeSet{p}}\leq\cardSet{\edgeSet{\lambda c}}=\mathcal{O}(N)\quad\Rightarrow\quad \nCoarse=\mathcal{O}(N)
    \end{equation}
    holds. Hence, the size of the coarse problem $\nCoarse$ only grows linearly with the number of subdomains $N$. This is related to efficiency, since the size of the coarse problem is a major factor for scalability because of the sequential computations necessary to evaluate the action of $\coarseMat^{-1}$.
\end{remark}
}
\begin{remark}
    To the best of our understanding, our procedure is very similar to the approach from \citeauthor{Yao_2012ab} in \cite{Yao_2012aa,Yao_2012ab} when considering a pure Dirichlet problem. Their construction is not explicitly described but both approaches shall eventually satisfy the same conditions. In contrast, we extend the underlying idea to problems with mixed boundary conditions by introducing additional edge weights and \revised{prove local invertibility for both the full Dirichlet and the more general case.}
\end{remark}

%% file: numerics.tex
In this section, we verify that the overall algorithm correctly solves general 3D magnetostatic problems and investigate some performance indicators. We focus on supporting our theoretical findings rather than high-performance computing. We employ the IGA-library \texttt{GeoPDEs} \cite{Vazquez_2016aa} in \texttt{Matlab} 2023a \cite{Mathworks_2023aa} for all tests. \revised{We obtain the actions of $\stiffMat_{\mathrm{rr}}^{-1}$ and $\coarseMat^{-1}$ on vectors via appropriate factorizations (e.g. $\mathbf{L}\mathbf{D}\mathbf{L}\trans$) provided by \texttt{Matlab}. The factorization of $\coarseMat$ in combination with forward/backward substitution can be exploited to solve \eqref{eq:intProb} iteratively using a Conjugate Gradient (CG) method without explicitly assembling $\mathbf{S}$.} After obtaining $\mults_{\mathrm{r}}$ from \eqref{eq:intProb}, one can recover $\dofs_{\mathrm{p}}$ and $\dofs_{\mathrm{r}}$ from \eqref{eq:coarseProb} and \eqref{eq:localProb}, respectively.

For tests related to numerical accuracy we use the method of manufactured solutions \cite[Section 6.3]{Oberkampf_2010aa}. We predefine the magnetic vector potential 
\begin{equation}
    \Afield_{\mathrm{ana}}(x,y,z)=
    \begin{bmatrix}
        \cos(y)\cos(z)\sin(x) \\
        -2\cos(x)\cos(z)\sin(y) \\
        \cos(x)\cos(y)\sin(z) \\
    \end{bmatrix}
    \quad\Rightarrow\quad
    \Bfield_{\mathrm{ana}}(x,y,z)=
    \Curl\Afield_{\mathrm{ana}}=
    \begin{bmatrix}
        -3\cos(x)\sin(y)\sin(z) \\
        0 \\
        3\cos(z)\sin(x)\sin(y) \\
    \end{bmatrix}\label{eq:anaSol}.
\end{equation}
Then, we compute the source current density and the boundary conditions \revised{by applying the respective operators and traces in \eqref{eq:vecPot1}-\eqref{eq:vecPot3} to $\Afield_{\mathrm{ana}}$ for a given domain configuration.} This allows us to study inhomogeneous Dirichlet and Neumann boundary conditions. For simplicity, we assume $\reluctivity\equiv1$.

Numerical errors will be measured by the subdomain-wise $\Hcurl{\Omega}$ seminorm \revised{of $\Afield$}
\begin{equation}
    \errBfield\coloneqq\sqrt{\sum_{i=1}^{N}\norm{\Bfield_{\mathrm{ana}}-\Bfield_{\mathrm{num}}}{L^2(\Omega^{(i)})}^2}\label{eq:errB},
\end{equation}
\revised{which corresponds to the $L^2$-error of the magnetic flux density $\Bfield$.} To measure the rate of convergence, we investigate the number of subdivisions in each direction in each subdomain $s_{\mathrm{h}}$, for which we may expect
\begin{equation}
    \errBfield\leq C h^{p}=\widetilde{C} s_{\mathrm{h}}^{-p}\label{eq:ErrEst}
\end{equation}
for suitable $C,\widetilde{C}\in\mathbb{R}_{>0}$.

Finally, the condition numbers of several matrices are investigated. We use the following definition \cite{Todd_1950aa}
\begin{equation*}
    \cond{\mathbf{X}}=\frac{\maxEigv{\mathbf{X}}}{\minEigv{\mathbf{X}}},
\end{equation*}
where $\maxEigv{\mathbf{X}}$ and $\minEigv{\mathbf{X}}$ are the maximal and minimal absolute eigenvalues of $\mathbf{X}$, respectively. \revised{In practice, these are estimated by an approach based on the Lanczos algorithm \cite[Sec.~6.7.3]{Saad_2000aa}.}

\subsection{Proof of Concept on Spherical Geometry}\label{subsec:sphere}
In this subsection, we provide experimental proof that our spanning tree construction and the IETI-DP algorithm work as promised. We use a spherical geometry, with radius 2, as depicted in \autoref{fig:sphExp_a}, and solve a manufactured problem with solution \eqref{eq:anaSol} for either Dirichlet $\Gamma_\mathrm{D}=\partial\Omega$ or Neumann boundary conditions $\Gamma_\mathrm{N}=\partial\Omega$.
We visualize the control mesh with tree and primal edges in \autoref{fig:sphExp_b} and the numerical flux density in \autoref{fig:sphRes_a}.
\begin{figure}
    \centering
    \begin{subfigure}[b]{0.49\linewidth}
        \centering
		\includegraphics[width=0.7\linewidth, trim= 6.2cm 1cm 6cm 1cm, clip]{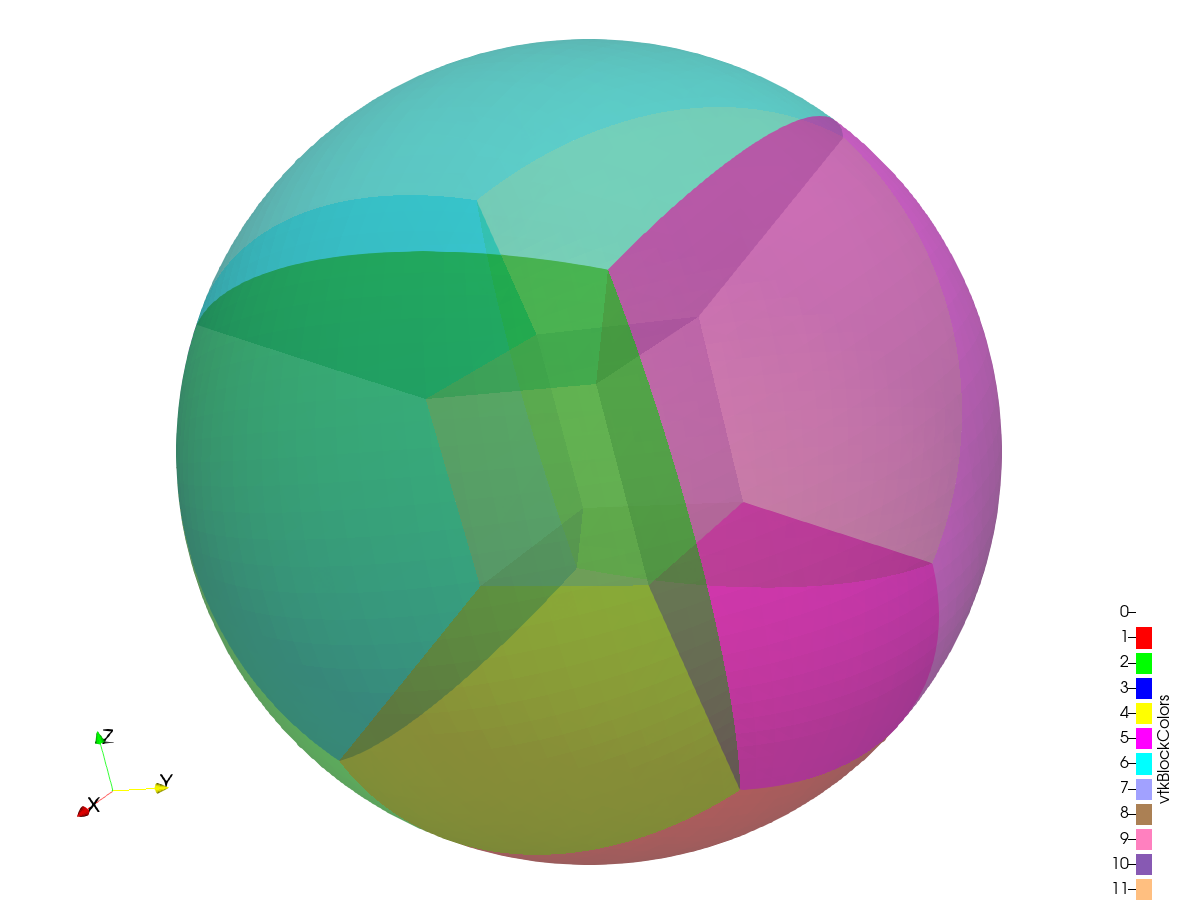}
        \caption{Decomposition of sphere in 7 subdomains.}
        \label{fig:sphExp_a}
    \end{subfigure}
    \begin{subfigure}[b]{0.49\linewidth}
        \centering
        \includegraphics[width=0.9\linewidth,height=5.5cm]{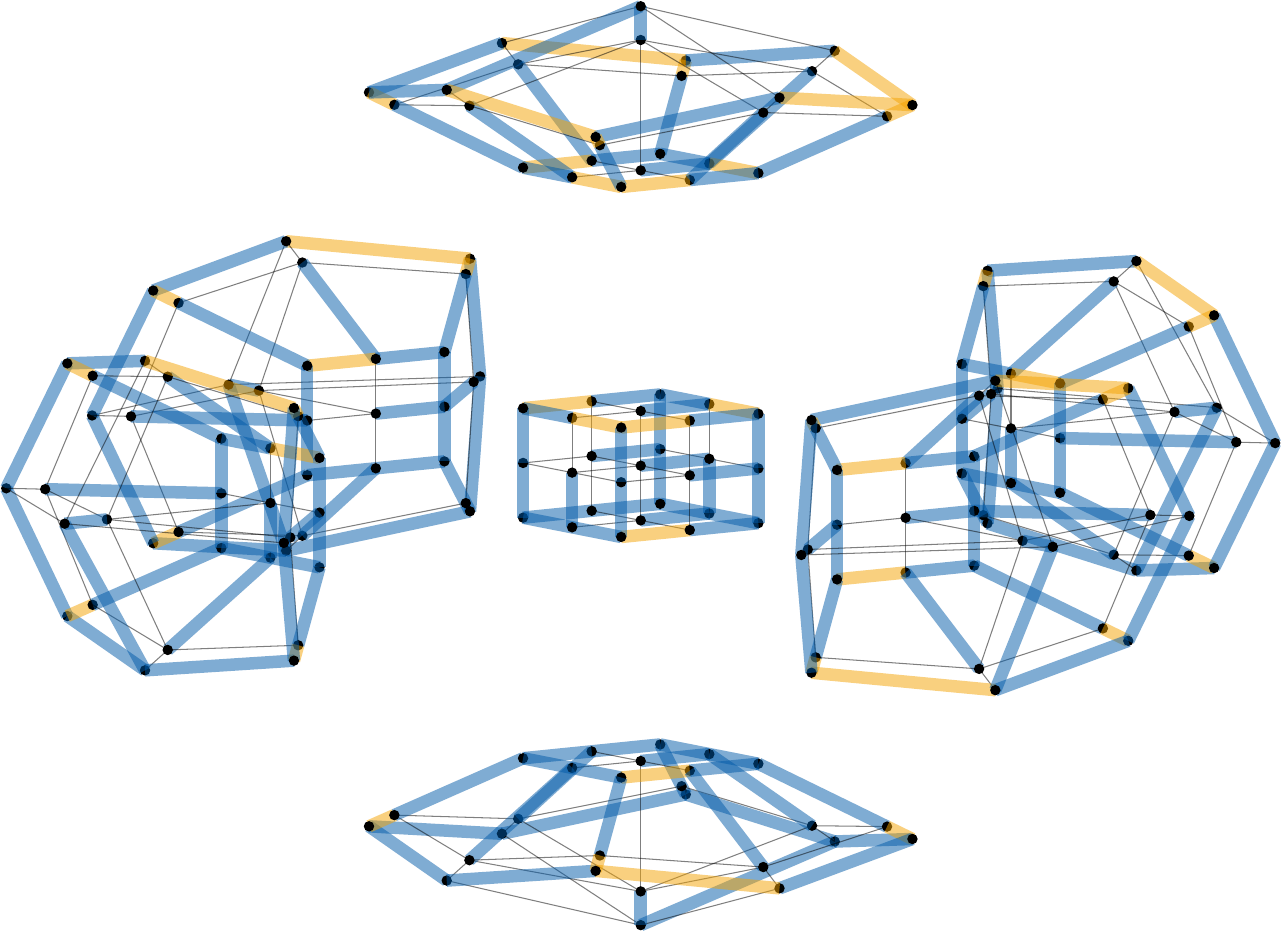}
        \caption{Control mesh with tree and primal DOFs.}
        \label{fig:sphExp_b}
    \end{subfigure}
    \caption{Geometry and exemplary control mesh for spherical geometry including tree DOFs (blue), and primal DOFs (orange).}
    \label{fig:sphExp}
\end{figure}
The main experiment of this subsection is a convergence study of $\errBfield$ with respect to $s_{\mathrm{h}}$. The results are visualized in \autoref{fig:sphRes_b} where one can observe the expected optimal behavior as given in \eqref{eq:ErrEst}. \revised{Furthermore, no issues due to singular matrices were observed implying that all $\stiffMat_{\mathrm{rr}}^{(i)}$ are gauged correctly.} Therefore, we conclude that the combination of our tree construction, the primal DOF selection and IETI-DP can be used to compute correct solutions for the flux density $\Bfield$ to \eqref{eq:vecPot1}-\eqref{eq:vecPot3}.
\begin{figure}
    \centering
    \begin{subfigure}[b]{0.4\linewidth}
        \centering
		\includegraphics[width=0.85\linewidth, trim= 6.2cm 1cm 6cm 1cm, clip]{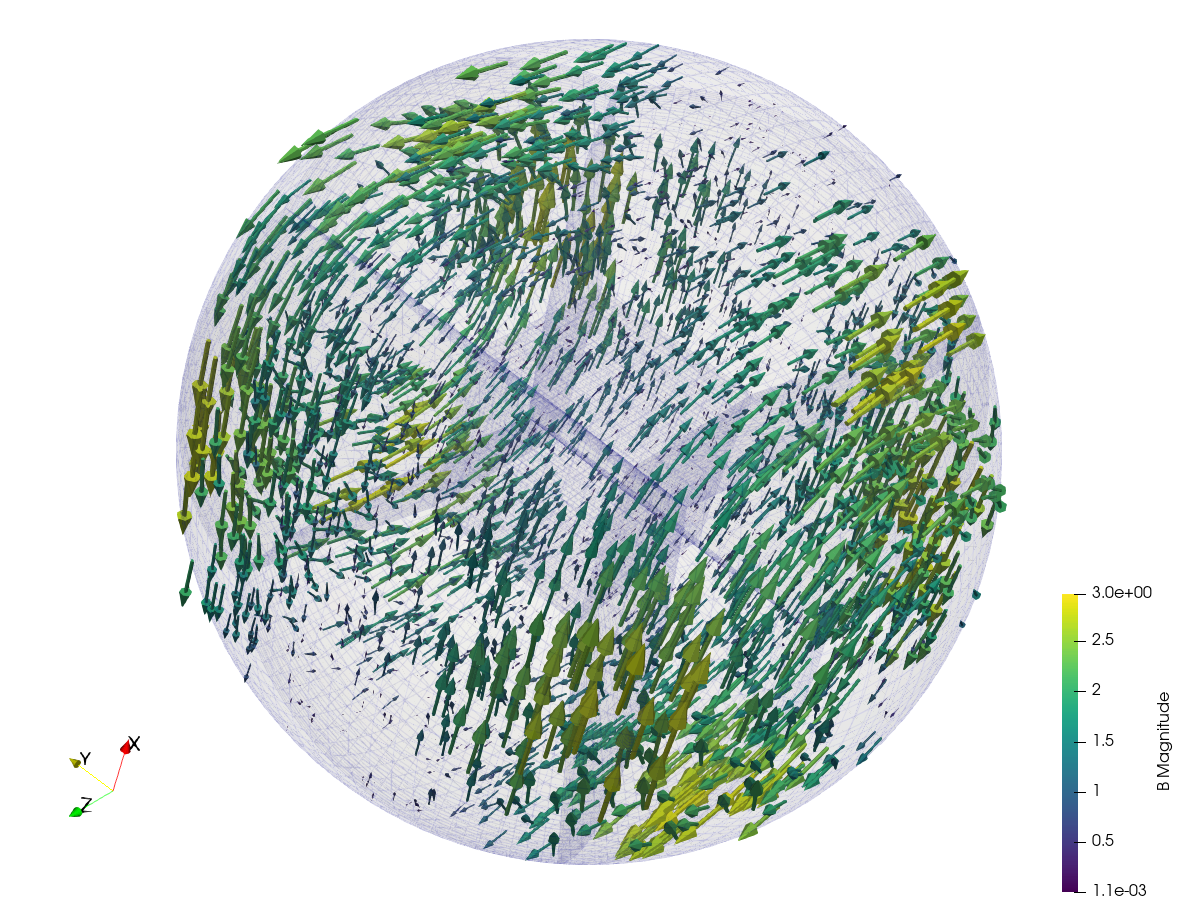}
        \caption{Numerically computed flux density.}
        \label{fig:sphRes_a}
    \end{subfigure}
    \begin{subfigure}[b]{0.59\linewidth}
        \centering
        \begin{tikzpicture}
			\begin{loglogaxis}[
			    xmin=1.5,xmax=20, ymin=1e-3, ymax=4.5, log origin=infty, xlabel={Subdivisions $s_{\mathrm{h}}$}, ylabel={Error $\errBfield$}, xtick={2,4,8,16}, xticklabels={2,4,8,16}, width=0.75\linewidth, height=6cm, legend style={at={(1.02,0.5)}, anchor=west, font=\tiny}, legend columns=1, legend cell align=left,
            ]
                \addplot+[mark=square, blue] table[x index = 0,y index = 1, col sep=comma] {data/nmnResultsSphere.csv};
				\addplot+[mark=square, red] table[x index = 0,y index = 2, col sep=comma] {data/nmnResultsSphere.csv};
				\addplot+[mark=square, brown] table[x index = 0,y index = 3, col sep=comma] {data/nmnResultsSphere.csv};

                \addplot+[mark=x, dashed, orange, every mark/.append style={solid}] table[x index = 0,y index = 1, col sep=comma] {data/dirResultsSphere.csv};
				\addplot+[mark=x, dashed, cyan, every mark/.append style={solid}] table[x index = 0,y index = 2, col sep=comma] {data/dirResultsSphere.csv};
				\addplot+[mark=x, dashed, green, every mark/.append style={solid}] table[x index = 0,y index = 3, col sep=comma] {data/dirResultsSphere.csv};

                \addplot+[black, dashed, mark=none, samples at={2,3,4,6,8,11,16}] {10^(0.8-log10(x))};
                \addplot+[black, dashed, mark=none, samples at={2,3,4,6,8,11,16}, forget plot] {10^(0.8-2*log10(x))};
                \addplot+[black, dashed, mark=none, samples at={2,3,4,6,8,11,16}, forget plot] {10^(0.8-3*log10(x))};

                \legend{$\text{Neu.,}~p=1$,$\text{Neu.,}~ p=2$,$\text{Neu.,}~p=3$,$\text{Dir.,}~p=1$,$\text{Dir.,}~p=2$,$\text{Dir.,}~p=3$,Order: $-p$};
			\end{loglogaxis}
        \end{tikzpicture}
        \caption{Error of Neumann (square) and Dirichlet (cross) problem.}
        \label{fig:sphRes_b}
    \end{subfigure}
    \caption{Visualization of solution and convergence study for Dirichlet and Neumann problem (with $\cgTol=\SI{1e-6}{}$ for the interface problem).}
    \label{fig:sphRes}
\end{figure}

\subsection{Experiments on Toroidal Geometry}\label{subsec:toroidal}
\revised{In this subsection, we examine the toroidal geometry in \autoref{fig:torExp_a} and the associated issues with the harmonic field $\harmonic$. This domain is not simply connected, i.e. its Betti number $b_1=1>0$ which means that there is a hole in $\Omega$.} We either employ full Neumann conditions, i.e. $\Gamma_\mathrm{N}=\partial\Omega$, or the setup with mixed boundary
\begin{equation*}
	\Gamma_{\mathrm{D}}\coloneqq\left\{r=r_{\mathrm{i}}, ~~\theta\in\left(0,2\pi\right],~~z\in\left(0,2\right)\right\} 
	\quad
	\text{and}
	\quad
	\Gamma_{\mathrm{N}}:=\partial\Omega\Setminus\closure{\Gamma}_{\mathrm{D}}.
\end{equation*}
We expect problems related to a non-trivial kernel to occur only for the full Neumann, but not in the mixed setting. When the Dirichlet constraints, which form a loop around the hole, are strongly imposed, the harmonic contribution to the kernel is automatically removed. However, in the full Neumann setting, we need to add an edge to the spanning tree that closes a loop around the hole \cite{Dlotko_2017aa}.

\revised{In practice, the predicted issue is not immediately visible: In \autoref{fig:torExp_b} optimal convergence behavior is observable in both the Neumann and the mixed setup. The reason is that \texttt{Matlab}'s $\mathbf{L}\mathbf{D}\mathbf{L}\trans$-factorization can cope with the singular systems, but issues warnings.} However, when examining the condition numbers, we see that $\cond{\coarseMat}$ \revised{cannot be computed with the employed method}, see \autoref{tab:torKernels} for $N=6$, in the full Neumann boundary case. This indicates a rank deficit. Nevertheless, $\cond{\interfaceMat}$ and the error $\errBfield$ remain comparable to the problem with mixed boundary conditions due to the implicit gauge of the applied factorizations. The configuration denoted with `Mod.' in \autoref{tab:torKernels} refers to the Neumann problem with an additional tree edge added due to cohomology considerations. As expected, this repairs the rank deficit and $\cond{\coarseMat}$ behaves normally. A visualization is given in \autoref{fig:torMeshSub6} where the violet edge is added to the blue tree. We include a second simpler decomposition with $N=3$ where the rank deficit can be observed in the local context, i.e. $\cond{\stiffMat_{\mathrm{rr}}}$ becomes very large for the Neumann setting. This may be surprising since we restrict harmonic fields to be only a global issue, but a closer look at \autoref{fig:torMeshSub3} explains this behavior: without the additional edge, the graph $\graph{e}^{(i)}$ of eliminated DOFs is not connected in one subdomain. However, as before, the \texttt{Matlab} solvers are able to deal with the rank deficit such that the corresponding $\cond{\coarseMat}$, $\cond{\interfaceMat}$ and $\errBfield$ stay appropriate. But again, adding an edge that closes a loop around the hole eliminates the additional kernel element and its associated difficulties. This verifies our assumptions and demonstrates that our algorithm can correctly handle both mixed boundary conditions and geometries with $b_1\neq 0$ if a little additional care is taken. We recommend using a cohomology-aware tree construction, such as \cite{Dlotko_2017aa}, to automatically create appropriate \revised{``belted'' trees that directly include additional loop-closing edges.}

\begin{figure}
    \centering
    \begin{subfigure}[b]{0.4\linewidth}
        \centering
        \includegraphics[width=0.85\linewidth, trim= 12cm 3cm 12cm 3cm, clip]{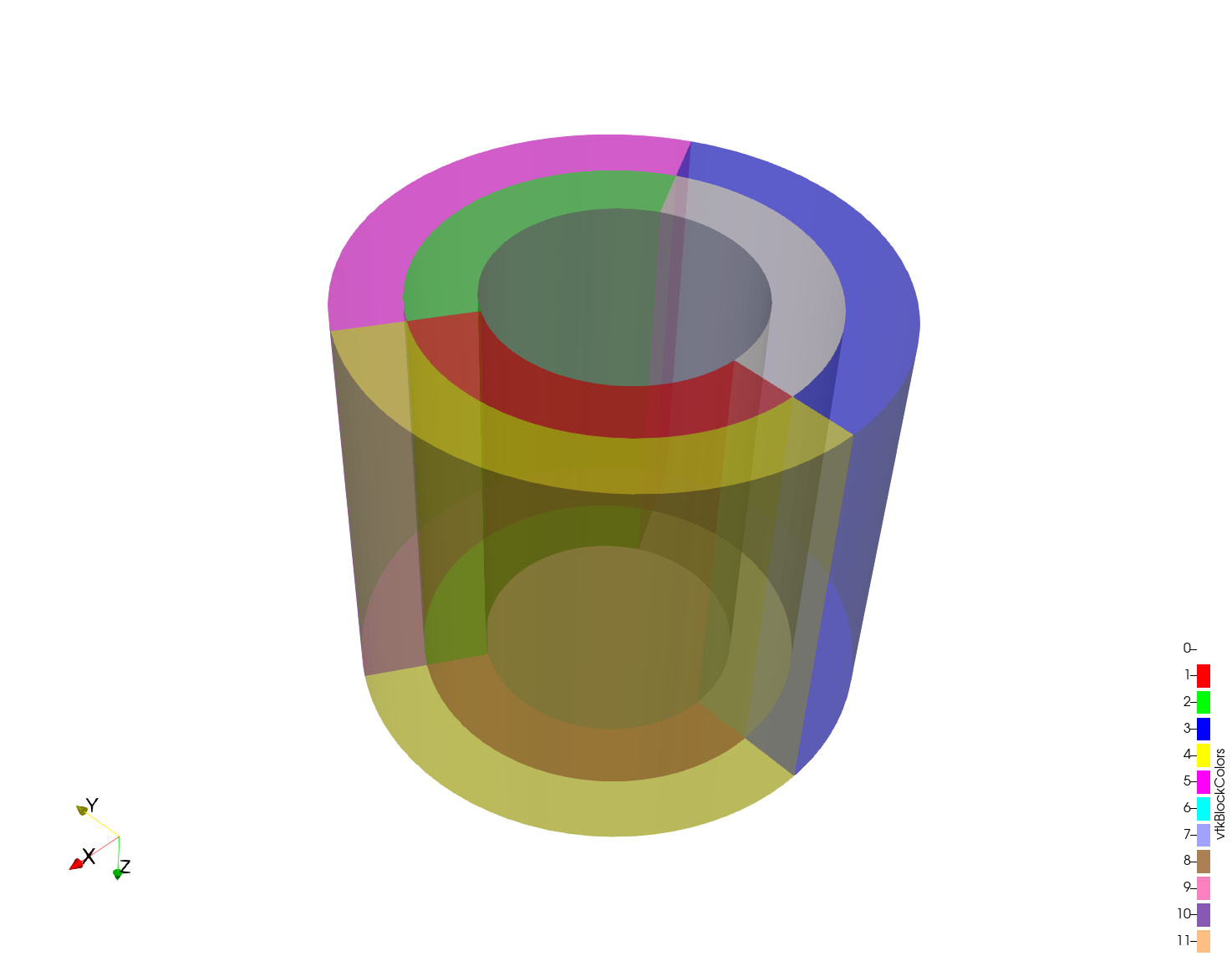}
        \caption{Decomposition of torus in 6 subdomains.}
        \label{fig:torExp_a}
    \end{subfigure}
    \begin{subfigure}[b]{0.59\linewidth}
        \centering
        \begin{tikzpicture}
			\begin{loglogaxis}[
			    xmin=1.5, xmax=20, ymin=1e-4, ymax=2, log origin=infty, xlabel={Subdivisions $s_{\mathrm{h}}$}, ylabel={Error $\errBfield$}, xtick={2,4,8,16}, xticklabels={2,4,8,16}, width=0.75\linewidth, height=6cm, legend style={at={(1.02,0.5)}, anchor=west, font=\tiny}, legend columns=1, legend cell align=left,
            ]

            \addplot+[mark=square, blue] table[x index = 0,y index = 1, col sep=comma] {data/torus_neu_errData.csv};
			\addplot+[mark=square, red] table[x index = 0,y index = 2, col sep=comma] {data/torus_neu_errData.csv};
			\addplot+[mark=square, brown] table[x index = 0,y index = 3, col sep=comma] {data/torus_neu_errData.csv};

            \addplot+[mark=x, dashed, orange, every mark/.append style={solid}] table[x index = 0,y index = 1, col sep=comma] {data/torus_dir_errData.csv};
			\addplot+[mark=x, dashed, cyan, every mark/.append style={solid}] table[x index = 0,y index = 2, col sep=comma] {data/torus_dir_errData.csv};
			\addplot+[mark=x, dashed, green, every mark/.append style={solid}] table[x index = 0,y index = 3, col sep=comma] {data/torus_dir_errData.csv};

            \addplot+[black, dashed, mark=none, samples at={2,3,4,6,8,16}] {10^(0.3-log10(x))};
            \addplot+[black, dashed, mark=none, samples at={2,3,4,6,8,16}, forget plot] {10^(0.3-2*log10(x))};
            \addplot+[black, dashed, mark=none, samples at={2,3,4,6,8,16}, forget plot] {10^(0.3-3*log10(x))};

            \legend{$\text{Neu.,}~p=1$,$\text{Neu.,}~ p=2$,$\text{Neu.,}~p=3$,$\text{Mix.,}~p=1$,$\text{Mix.,}~p=2$,$\text{Mix.,}~p=3$,Order: $-p$}; 
            
			\end{loglogaxis}
        \end{tikzpicture}
        \caption{Convergence for mixed BCs and full Neumann problem.}
        \label{fig:torExp_b}
    \end{subfigure}
    \caption{Visualization of model problem with geometry (height: $h=2$, inner radius: $r_{\mathrm{i}}=1$, outer radius: $r_{\mathrm{o}}=2$) and convergence study for a configuration with Dirichlet boundary conditions around the hole and a configuration with full Neumann boundary conditions (with $\cgTol=\SI{1e-6}{}$ for the interface problem).}
    \label{fig:torExp}
\end{figure}
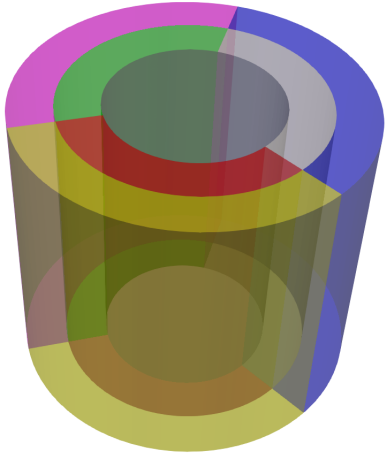
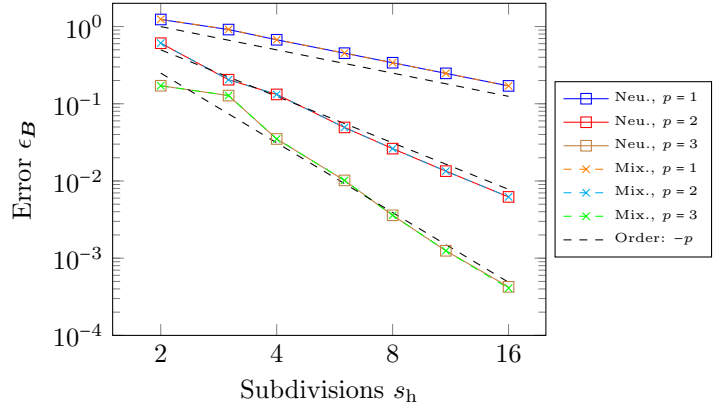

\begin{table}
    \centering
    \begin{filecontents*}{data/ietiTorusHarmonicResults.csv}
deg,sh,N,Exp.,arr,p,lam,condArr,condF,noPrec condS,power noPrec condS,noPrec Iter,lumped condS,power lumped condS,lumped Iter,dir condS,power dir condS,dir Iter,err
2,8,3,Mix.,4854,3,240,1.95e+04,1.00e+00,6.42e+02,6.46e+02,64,3.68e+02,4.02e+02,49,1.00e+00,1.00e+00,3,2.63e-02
2,8,3,Neu.,5098,3,240,411994468948102,1.00e+00,8.16e+02,8.18e+02,63,4.37e+02,4.93e+02,51,1,1.00e+00,2,2.63e-02
2,8,3,Mod.,5097,3,240,6.32e+04,1.00e+00,8.15e+02,8.17e+02,62,4.37e+02,4.93e+02,50,1,1.00e+00,2,2.63e-02
2,8,6,Mix.,9942,12,720,1.40e+05,1.15e+01,1.78e+03,1.93e+03,120,1.16e+03,1.31e+03,96,9.52e+00,9.55e+00,15,2.62e-02
2,8,6,Neu.,10188,10,720,2.82e+05,NaN,1.69e+03,1.70e+03,118,1.20e+03,1.36e+03,98,2.99e+00,3.00e+00,11,2.62e-02
2,8,6,Mod.,10187,10,720,2.82e+05,1.28e+01,1.69e+03,1.70e+03,117,1.20e+03,1.36e+03,97,3.00e+00,3.01e+00,11,2.62e-02
\end{filecontents*}
\pgfplotstabletypeset[
    	col sep=comma,
        columns={Exp.,N,condArr,condF,noPrec condS,err},
    	columns/Exp./.style={column name={Config.}, string type},
    	columns/N/.style={column name=$N$},
        columns/condArr/.style={column name=$\cond{\stiffMat_{\mathrm{rr}}}$, , /pgf/number format/.cd, sci, sci zerofill, precision=2},
        columns/condF/.style={column name=$\cond{\coarseMat}$, /pgf/number format/.cd, sci, sci zerofill, precision=2},
        columns/noPrec condS/.style={column name=$\cond{\interfaceMat}$, /pgf/number format/.cd, sci, sci zerofill, precision=2},
    	columns/err/.style={column name=$\errBfield$, /pgf/number format/.cd, sci, sci zerofill, precision=2},
        specialCell1/.style={@cell content=\ensuremath{\mathbf{#1}}},
        specialCell2/.style={@cell content={---}},
        every head row/.style={before row=\hline,after row=\hline},
        every row no 2/.style={after row=\hline},
        every last row/.style={after row=\hline},
        every row 1 column 2/.style={
            postproc cell content/.append code={%
                \pgfkeysgetvalue{/pgfplots/table/@cell content}{\myTmpVal}%
                \pgfkeysalso{specialCell1/.expand once={\myTmpVal}}
            }%
        },
        every row 4 column 3/.style={
            postproc cell content/.append code={%
                \pgfkeysgetvalue{/pgfplots/table/@cell content}{\myTmpVal}%
                \pgfkeysalso{specialCell2/.expand once={\myTmpVal}}
            }%
        },
    ]{data/ietiTorusHarmonicResults.csv}
    \caption{Examination of condition numbers and error for different boundary configurations and a modified gauge for the Neumann problem ($p=2$, $s_{\mathrm{h}}=8$ and $\cgTol=\SI{1e-6}{}$). The additional harmonic kernel introduces numerical difficulties highlighted in bold or with ``---'' which indicates that the condition number estimation terminated early due to failing internal checks.}
    \label{tab:torKernels}
\end{table}

\begin{figure}
    \centering
    \begin{subfigure}[b]{0.49\linewidth}
        \centering
        \includegraphics[width=0.9\linewidth]{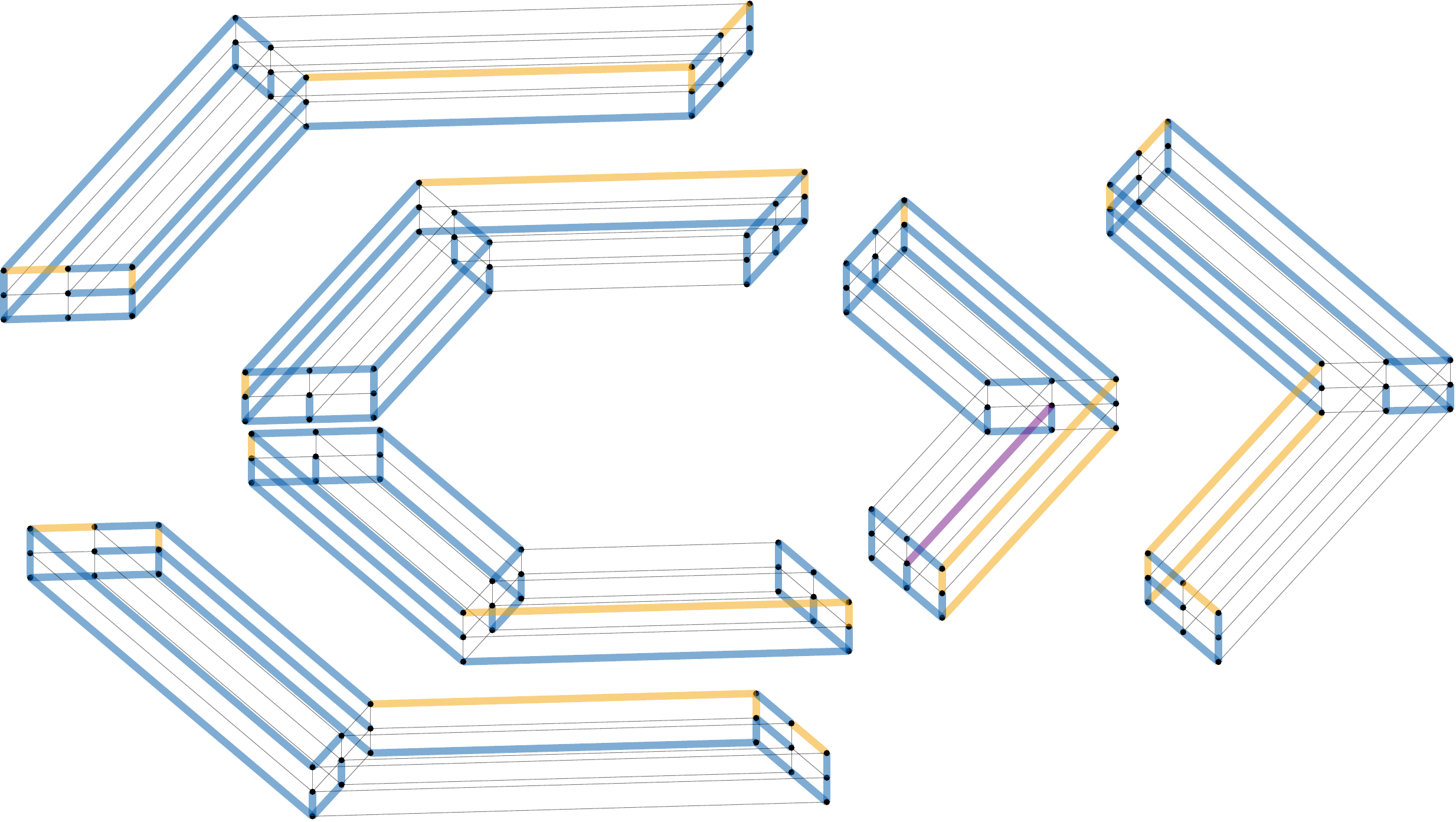}
        \caption{More complex decomposition with 6 subdomains.}
        \label{fig:torMeshSub6}
    \end{subfigure}
    \begin{subfigure}[b]{0.49\linewidth}
        \centering
        \includegraphics[width=0.9\linewidth]{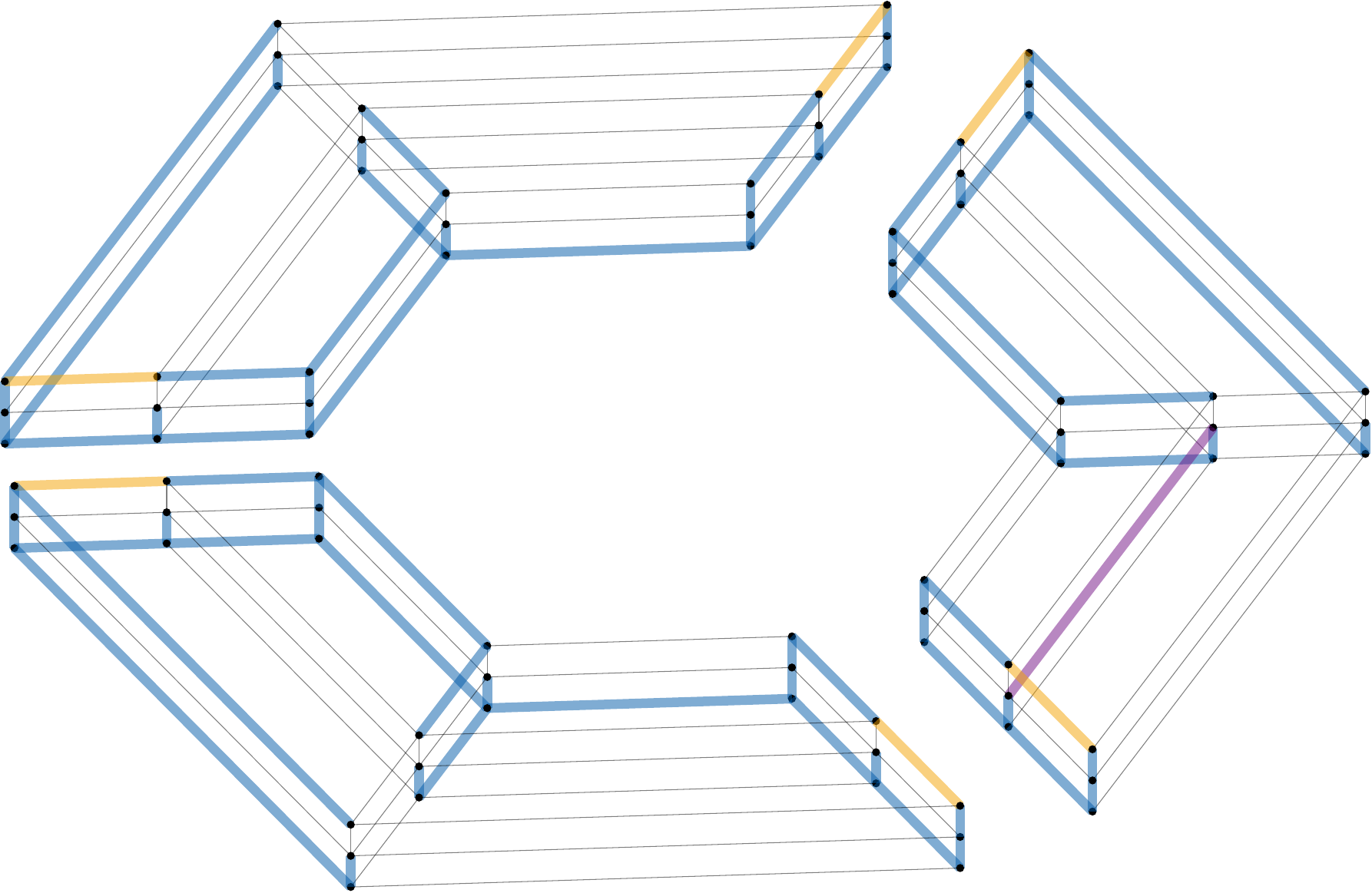}
        \caption{Simple decomposition with 3 subdomains.}
        \label{fig:torMeshSub3}
    \end{subfigure}
    \caption{Different decompositions of the torus, represented by their local meshes for $p=1$ and $s_{\mathrm{h}}=2$. The meshes include the tree edges (blue), the primal edges (orange) and the additional tree edge used for removing the harmonic kernel (violet).}
    \label{fig:torMesh}
\end{figure}

\subsection{Numerical Scalability}\label{subsec:scalab}
To evaluate the behavior of our algorithm for large problems, we follow the tests and explanations of \citeauthor{Farhat_2001aa} in \cite{Farhat_2001aa} and references therein. They determine scalability by looking at the sequential bottlenecks in the dual-primal algorithm. In the context of parallel performance, the sequential part is relevant in both Amdahl's law (strong scalability) \cite{Amdahl_1967aa,Klawonn_2010aa} and Gustafson's law (weak scalability) \cite{Gustafson_1998aa,Klawonn_2010aa}.

For the following test, we use the unit cube geometry with side length 1, which is decomposed into $N$ equally sized cubes. For the boundary, we employ mixed boundary conditions with a not connected Dirichlet boundary
\begin{equation*}
	\Gamma_{\mathrm{D}}\coloneqq\bigl\{[x_1, x_2, x_3]^{\top},~~x_1,x_3\in(0,1),~x_2\in\{0,1\}\bigr\},
	\quad
	\text{and}
	\quad
	\Gamma_{\mathrm{N}}:=\partial\Omega\Setminus\closure{\Gamma}_{\mathrm{D}}.
\end{equation*}
which is handled automatically by our algorithm. 
Due to this special decomposition, we can introduce the coarse problem size (body diagonal of each subdomain) as $H=\sqrt{3}N^{-\nicefrac{1}{3}}$ and the corresponding number of coarse subdivisions $s_{\mathrm{H}}=\sqrt{3}H^{-1}$. If the local subdivisions are uniform in each subdomain, the relation $h=\sqrt{3}(s_{\mathrm{h}}s_{\mathrm{H}})^{-1}$ can be derived. Consequently, the well-known FETI-ratio, as discussed in  \cite{Farhat_2001aa} and references therein, can be expressed as $\nicefrac{H}{h}=s_{\mathrm{h}}$. The quantities $s_{\mathrm{h}}$ and $s_{\mathrm{H}}$ allow a concise discussion of the following three test configurations:
\begin{enumerate}
    \item The first test consists of setting $H=\mathrm{const.}$, i.e., $s_{\mathrm{H}}=\mathrm{const.}$ and varying $s_{\mathrm{h}}$. In other words, we only change the size of the local problems.
    \item The second test consists of setting $h=\mathrm{const.}$, i.e., $s_{\mathrm{h}}s_{\mathrm{H}}=\mathrm{const.}$ In other words, we keep the overall problem size constant while increasing the number of subdomains (i.e., processors). This is equivalent to testing for so-called strong scalability, which follows Amdahl's law.
    \item In the third test, we set $\nicefrac{H}{h}=s_{\mathrm{h}}=\mathrm{const.}$ and vary $s_{\mathrm{H}}$. \revised{In other words,} we increase both, the problem size and the number of subdomains/processors, while keeping the workload per processor constant. This is a test of the so-called weak scalability based on Gustafson's law.
\end{enumerate}
The results of the first scalability test are summarized in \autoref{tab:scaV1}. We obtain the expected convergence orders w.r.t. $s_{\mathrm{h}}$ as in the other numerical tests of this section. \revised{We further observe that $\cond{\interfaceMat}$ deteriorates fast with increasingly refined local discretizations. Applying the lumped preconditioner $\mathbf{M}_{\mathrm{L}}^{-1}$ improves the conditioning already but a better perfomance is shown by the Dirichlet preconditioner $\mathbf{M}_{\mathrm{D}}^{-1}$. This numerical test even suggests that the relation
\begin{equation}
    \cond{\mathbf{M}_{\mathrm{D}}^{-1}\interfaceMat}\leq C\left(\log_{10}(s_{\mathrm{h}})+1\right)^2 = C\left(\log_{10}\left(\nicefrac{H}{h}\right)+1\right)^2\label{eq:condBound}
\end{equation}
as in \cite{Mandel_2001aa,Klawonn_2002aa} holds. A visualization is provided in \autoref{fig:scalIndicators_a} where the dashed, red function was only added as a visual comparison. At last, we see that the size of $\mathbf{p}$, i.e. $\nCoarse$, remains constant when the subdomains are refined locally which agrees with our considerations in \autoref{rem:effPrimal}.}

\begin{table}
    \centering
    \begin{filecontents*}{data/ietiTest1Results_iter.csv}
deg,sH,sh,arr,p,lam,noPrec condS,power noPrec condS,noPrec Iter,lumped condS,power lumped condS,lumped Iter,dir condS,power dir condS,dir Iter,err,a
1,2,2,168,12,36,1.20e+01,1.20e+01,22,3.78e+00,3.98e+00,13,1.62e+00,1.62e+00,7,2.62e+00,432
1,2,4,1256,12,180,6.22e+01,6.37e+01,41,2.13e+01,2.25e+01,27,2.28e+00,2.28e+00,9,1.33e+00,2400
1,2,8,9192,12,756,2.79e+02,3.10e+02,72,1.12e+02,1.20e+02,50,3.15e+00,3.15e+00,11,6.69e-01,15552
1,2,16,69608,12,3060,1.08e+03,1.41e+03,124,5.07e+02,5.66e+02,93,4.21e+00,4.21e+00,14,3.35e-01,110976
\end{filecontents*}
\pgfplotstabletypeset[
        col sep=comma,
        columns={sH,sh,p,noPrec condS,lumped condS,dir condS,noPrec Iter,lumped Iter,dir Iter,err},
        columns/sH/.style={column name=$s_{\mathrm{H}}$},
        columns/sh/.style={column name=$s_{\mathrm{h}}$},
        columns/p/.style={column name=$\nCoarse$, /pgf/number format/.cd, fixed, fixed zerofill, precision=0, 1000 sep={}},
        columns/noPrec condS/.style={column name=$\cond{\interfaceMat}$, /pgf/number format/.cd, sci, sci zerofill, precision=2},
        columns/noPrec Iter/.style={column name={$\mathrm{it}(\interfaceMat)$}},
        columns/lumped condS/.style={column name=$\cond{\mathbf{M}_{\mathrm{L}}^{-1}\interfaceMat}$, /pgf/number format/.cd, sci, sci zerofill, precision=2},
        columns/lumped Iter/.style={column name={$\mathrm{it}(\mathbf{M}_{\mathrm{L}}^{-1}\interfaceMat)$}},
        columns/dir condS/.style={column name=$\cond{\mathbf{M}_{\mathrm{D}}^{-1}\interfaceMat}$, /pgf/number format/.cd, sci, sci zerofill, precision=2},
        columns/dir Iter/.style={column name={$\mathrm{it}(\mathbf{M}_{\mathrm{D}}^{-1}\interfaceMat)$}},
        columns/err/.style={column name=$\errBfield$, /pgf/number format/.cd, sci, sci zerofill, precision=2},
        every head row/.style={before row=\hline,after row=\hline},
        every last row/.style={after row=\hline},   
    ]{data/ietiTest1Results_iter.csv}
    \caption{First scalability experiment with $s_{\mathrm{H}}=2$, i.e., $N=8$ on cube ($p=1$ and $\cgTol=\SI{1e-6}{}$). The value $\mathrm{it}(\bullet)$ refers to the number of PCG iterations with the respective (preconditioned) system matrix $\bullet$.}
    \label{tab:scaV1}
\end{table}

\begin{figure}
    \centering
    \begin{subfigure}[c]{0.49\linewidth}
        \centering
        \begin{tikzpicture}
            \begin{axis}[
                xmin=1, xmax=33, xlabel={Subdivisions $s_{\mathrm{h}}$}, ylabel={Condition Number}, xtick={2,4,8,16,32}, xticklabels={2,4,8,16,32}, width=0.9\linewidth, height=6cm, font=\small, legend pos = south east, legend style = {font=\scriptsize}, legend cell align=left
            ]
                \addplot+[mark=+,blue] table[x = {sh},y = {pcg eigest.}, col sep=comma] {data/ietiCurlConditioning.csv};

                \addplot+[mark=none,black,dashed,domain=2:32,samples=101] {(log10(x)+1)^2};
                \addplot+[mark=none,red,dashed,domain=2:32,samples=101] {1.8*(log10(x)+1)};

                \legend{$\cond{\mathbf{M}_{\mathrm{D}}^{-1}\interfaceMat}$,$\left(\log_{10}(s_{\mathrm{h}})+1\right)^2$,$1.8\left(\log_{10}(s_{\mathrm{h}})+1\right)$};
            \end{axis}
        \end{tikzpicture}
        \caption{Condition number with bound.}
        \label{fig:scalIndicators_a}
    \end{subfigure}
    \begin{subfigure}[c]{0.49\linewidth}
        \centering
        \begin{tikzpicture}
            \begin{loglogaxis}[
                xmin=4, xmax=8000, ymin=1, ymax=2e4, xlabel={Subdomains $N=s_{\mathrm{H}}^3$}, ylabel={Size of Coarse Problem}, width=0.9\linewidth, height=6cm, font=\small, legend pos = south east, legend style = {font=\scriptsize}, legend cell align=left, 
            ]
                \addplot+[mark=x,blue] table[x expr= {\thisrow{sH}^3},y = {p}, col sep=comma] {data/ietiTest2Results_iter.csv};
                \addplot+[mark=none,dashed,black,domain=8:4096,samples=200] {3*x};

                \legend{$\nCoarse$,$\mathcal{O}(N)$};
            \end{loglogaxis}
        \end{tikzpicture}
        \caption{Size of coarse problem with bound.}
        \label{fig:scalIndicators_b}
    \end{subfigure}
    \caption{Visualization of scalability indicators condition number of preconditioned interface problem and size of coarse problem.}
    \label{fig:scalIndicators}
\end{figure}

The results of the second scalability experiment (strong scaling) are shown in \autoref{tab:scaV2}. \revised{We keep the overall problem size constant while increasing the number of subdomains, i.e. we coarsen the local discretizations accordingly.} Therefore, it is not surprising that we observe the same error in each simulation run \revised{which confirms that a local refinement can be exchanged with a corresponding coarse refinement. For the size of the coarse problem $\nCoarse$, we observe a linear increase with respect to the number of subdomains $N=s_{\mathrm{H}}^3$ which is visualized in \autoref{fig:scalIndicators_b}. This verifies our considerations in \autoref{rem:effPrimal} and $\nCoarse=\mathcal{O}(N)$. The size of the coarse problem is further explored in \autoref{app:Ratio} where we compute relevant numbers, e.g. $\cardSet{\edgeSet{\lambda c}}$, analytically for the configuration of this subsection and compare them to the overall number of DOFs $n$. For the condition number and the amount of PCG iterations with different preconditioners, we observe a decrease as the number of local subdivisions decreases. This already suggests that the constant $C>0$ in \eqref{eq:condBound} does not depend on $s_{\mathrm{H}}$.}

\begin{table}
    \centering
    \begin{filecontents*}{data/ietiTest2Results_iter.csv}
deg,sH,sh,arr,p,lam,noPrec condS,power noPrec condS,noPrec Iter,lumped condS,power lumped condS,lumped Iter,dir condS,power dir condS,dir Iter,err,a
1,2,16,69608,12,3060,1.08e+03,1.41e+03,124,5.07e+02,5.65e+02,93,4.21e+00,4.21e+00,14,3.35e-01,110976
1,4,8,75504,128,9072,3.70e+02,3.83e+02,102,1.35e+02,1.48e+02,68,3.95e+00,3.96e+00,13,3.35e-01,124416
1,8,4,85664,1056,20160,7.77e+01,8.14e+01,55,2.57e+01,2.76e+01,34,2.94e+00,2.98e+00,10,3.35e-01,153600
1,16,2,92736,8384,34560,1.26e+01,2.20e+01,23,4.65e+00,5.25e+00,14,1.99e+00,2.00e+00,8,3.35e-01,221184
\end{filecontents*}
\pgfplotstabletypeset[
        col sep=comma,
        columns={sH,sh,p,noPrec condS,lumped condS,dir condS,noPrec Iter,lumped Iter,dir Iter,err},
        columns/sH/.style={column name=$s_{\mathrm{H}}$},
        columns/sh/.style={column name=$s_{\mathrm{h}}$},
        columns/p/.style={column name=$\nCoarse$, /pgf/number format/.cd, fixed, fixed zerofill, precision=0, 1000 sep={}},
        columns/noPrec condS/.style={column name=$\cond{\interfaceMat}$, /pgf/number format/.cd, sci, sci zerofill, precision=2},
        columns/noPrec Iter/.style={column name={$\mathrm{it}(\interfaceMat)$}},
        columns/lumped condS/.style={column name=$\cond{\mathbf{M}_{\mathrm{L}}^{-1}\interfaceMat}$, /pgf/number format/.cd, sci, sci zerofill, precision=2},
        columns/lumped Iter/.style={column name={$\mathrm{it}(\mathbf{M}_{\mathrm{L}}^{-1}\interfaceMat)$}},
        columns/dir condS/.style={column name=$\cond{\mathbf{M}_{\mathrm{D}}^{-1}\interfaceMat}$, /pgf/number format/.cd, sci, sci zerofill, precision=2},
        columns/dir Iter/.style={column name={$\mathrm{it}(\mathbf{M}_{\mathrm{D}}^{-1}\interfaceMat)$}},
        columns/err/.style={column name=$\errBfield$, /pgf/number format/.cd, sci, sci zerofill, precision=2},
        every head row/.style={before row=\hline,after row=\hline},
        every last row/.style={after row=\hline},   
    ]{data/ietiTest2Results_iter.csv}
    \caption{Second scalability experiment with $s_{\mathrm{h}}s_{\mathrm{H}}=32$ on cube ($p=1$ and $\cgTol=\SI{1e-6}{}$). The value $\mathrm{it}(\bullet)$ refers to the number of PCG iterations with the respective (preconditioned) system matrix $\bullet$.}
    \label{tab:scaV2}
\end{table}

\revised{The results of the third scalability experiment (weak scaling) shown in \autoref{tab:scaV3} support the findings of the previous scalability experiments. The same error occurs as in experiment 1, which implies again that a local refinement can be exchanged with a corresponding coarse refinement. The observations of the condition number and PCG iterations of the (preconditioned) interface problem indicate again that the number of subdomains has no or at most a minor influence. Following \cite{Farhat_2001aa}, this implies that our method is weakly (there called numerically) scalable. Only the solving of the coarse problem, although its size scales linearly with the number of subdomains, may remain as a bottleneck for large problems which motivates using inexact TI approaches \cite{Klawonn_2007aa} where \eqref{eq:coarseProb} is solved iteratively.

In summary, we observed the typical behavior for scalable TI methods w.r.t the preconditioned interface problem. Furthermore, we demonstrated that the size of the coarse problem only grows linearly with the number of subdomains and estimated its contribution to the overall system.}

\begin{table}
    \centering
    \begin{filecontents*}{data/ietiTest3Results_iter.csv}
deg,sH,sh,arr,p,lam,noPrec condS,power noPrec condS,noPrec Iter,lumped condS,power lumped condS,lumped Iter,dir condS,power dir condS,dir Iter,err,a
1,2,2,168,12,36,1.20e+01,1.20e+01,22,3.78e+00,3.99e+00,13,1.62e+00,1.62e+00,7,2.62e+00,432
1,4,2,1392,128,432,1.39e+01,1.90e+01,25,4.78e+00,4.99e+00,15,1.90e+00,1.94e+00,8,1.33e+00,3456
1,8,2,11424,1056,4032,1.29e+01,2.14e+01,24,4.76e+00,5.19e+00,14,1.97e+00,1.99e+00,8,6.69e-01,27648
1,16,2,92736,8384,34560,1.26e+01,2.20e+01,23,4.65e+00,5.25e+00,14,1.99e+00,2.00e+00,8,3.35e-01,221184
\end{filecontents*}
\pgfplotstabletypeset[
        col sep=comma,
        columns={sH,sh,p,noPrec condS,lumped condS,dir condS,noPrec Iter,lumped Iter,dir Iter,err},
        columns/sH/.style={column name=$s_{\mathrm{H}}$},
        columns/sh/.style={column name=$s_{\mathrm{h}}$},
        columns/p/.style={column name=$\nCoarse$, /pgf/number format/.cd, fixed, fixed zerofill, precision=0, 1000 sep={}},
        columns/noPrec condS/.style={column name=$\cond{\interfaceMat}$, /pgf/number format/.cd, sci, sci zerofill, precision=2},
        columns/noPrec Iter/.style={column name={$\mathrm{it}(\interfaceMat)$}},
        columns/lumped condS/.style={column name=$\cond{\mathbf{M}_{\mathrm{L}}^{-1}\interfaceMat}$, /pgf/number format/.cd, sci, sci zerofill, precision=2},
        columns/lumped Iter/.style={column name={$\mathrm{it}(\mathbf{M}_{\mathrm{L}}^{-1}\interfaceMat)$}},
        columns/dir condS/.style={column name=$\cond{\mathbf{M}_{\mathrm{D}}^{-1}\interfaceMat}$, /pgf/number format/.cd, sci, sci zerofill, precision=2},
        columns/dir Iter/.style={column name={$\mathrm{it}(\mathbf{M}_{\mathrm{D}}^{-1}\interfaceMat)$}},
        columns/err/.style={column name=$\errBfield$, /pgf/number format/.cd, sci, sci zerofill, precision=2},
        every head row/.style={before row=\hline,after row=\hline},
        every last row/.style={after row=\hline},   
    ]{data/ietiTest3Results_iter.csv}
    \caption{Third scalability experiment with $s_{\mathrm{h}}=2$ on cube ($p=1$ and $\cgTol=\SI{1e-6}{}$). The value $\mathrm{it}(\bullet)$ refers to the number of PCG iterations with the respective (preconditioned) system matrix $\bullet$.}
    \label{tab:scaV3}
\end{table}

%% file: conclusion.tex
In this work, we discussed dual-primal Tearing and Interconnecting approach for magnetostatic problems using IGA edge-elements and mixed boundary conditions in combination with the tree-cotree gauge. In this context, we introduced an explicit way to construct appropriate trees and how to select primal DOFs to obtain local invertibility to enable parallel computations for each subdomain. Furthermore, we proved that the local subdomain contributions are invertible \revised{when using our proposed global tree in combination with our choice for the primal DOFs. Extensions to our work are for example considerations on and numerical experiments with inhomogeneous, possibly discontiuous material distributions as well as preconditioning in this context. Additionally, the application to industrial problems and an extension to eddy current problems are of interest.}

%% file: appendix.tex
\section{Ratio between Coarse and Full Problem}\label{app:Ratio}
\revised{In this appendix, we examine the size of the coarse problem $\nCoarse$ and how it relates to the full number of DOFs $n=\sum_{i=1}^{N}n_i$ for the configuration of \autoref{subsec:scalab}.} First, we look at a single subdomain $\Omega^{(i)}$ with $s_{\mathrm{h}}$ subdivisions in every direction for $p=1$. Note that the following concepts can be extended to settings with $p>1$ by adding a constant offset to $s_{\mathrm{h}}$. The wire basket of $\Omega^{(i)}$ consists of $m_{\mathrm{e}}^{(i)}=12$ macroscopic edges and $m_{\mathrm{n}}^{(i)}=8$ vertices, each shared by three macroscopic edges. One can easily check that one of these edges contains $s_{\mathrm{h}}$ element edges and $s_{\mathrm{h}}+1$ nodes which yields
\begin{align*}
    \cardSet{\nodeSet{\lambda}^{(i)}}&=12s_{\mathrm{h}}-4 = 12\left(s_{\mathrm{h}} + 1\right) - 16 = m_{\mathrm{e}}^{(i)}\left(s_{\mathrm{h}} + 1\right) - (3-1)m_{\mathrm{n}}^{(i)} \\
    \cardSet{\edgeSet{\lambda}^{(i)}} &= 12s_{\mathrm{h}} = m_{\mathrm{e}}^{(i)}s_{\mathrm{h}}
\end{align*}
for the edges and nodes of the corresponding wire basket graph $\graph{\lambda}^{(i)}$. If we look at the broader picture with $N=s_{\mathrm{H}}^3$ subdomains, we can compute the number of macroscopic edges by
\begin{equation*}
    m_{\mathrm{e}}=\left.\cardSet{\edgeSet{}}\right\vert_{s_{\mathrm{h}}=1}=\left.3s_{\mathrm{h}}s_{\mathrm{H}}\left(s_{\mathrm{h}}s_{\mathrm{H}} + 1\right)^2\right\vert_{s_{\mathrm{h}}=1}=3s_{\mathrm{H}}\left(s_{\mathrm{H}} + 1\right)^2,
\end{equation*}
which implies that
\begin{equation}
    \cardSet{\edgeSet{\lambda}}=m_{\mathrm{e}}s_{\mathrm{h}}=3s_{\mathrm{h}}s_{\mathrm{H}}\left(s_{\mathrm{H}} + 1\right)^2
\end{equation}
is the number of edges on $\graph{\lambda}$. Correspondingly, we can state
\begin{equation}
    \cardSet{\nodeSet{\lambda}}=m_{\mathrm{e}}\left(s_{\mathrm{h}}+1\right)-\alpha m_{\mathrm{n}}=m_{\mathrm{e}}s_{\mathrm{h}} + \underbrace{(m_{\mathrm{e}} - \alpha)m_{\mathrm{n}}}_{=c}=\cardSet{\edgeSet{\lambda}} + c
\end{equation}
for the number of nodes on the global wire basket where $\alpha$ and $c$ are appropriate constants. For the global number of endpoints
\begin{equation*}
    m_{\mathrm{n}}=\left.\cardSet{\nodeSet{}}\right\vert_{s_{\mathrm{h}}=1}=\left.\left(s_{\mathrm{h}}s_{\mathrm{H}} + 1\right)^3\right\vert_{s_{\mathrm{h}}=1}=\left(s_{\mathrm{H}} + 1\right)^3
\end{equation*}
can be calculated. The constant $\alpha$ only depends on how many edges are connected to each respective endpoint, which is quite technical to compute but clearly independent of the local discretization. To compute $c$, we can exploit that $\cardSet{\nodeSet{\lambda}}=m_{\mathrm{n}}$ needs to be valid for $s_{\mathrm{h}}=1$. This results in
\begin{equation*}
    m_{\mathrm{n}}=\cardSet{\nodeSet{\lambda}}=m_{\mathrm{e}} - c ~~\Leftrightarrow ~~ c = 3s_{\mathrm{H}}\left(s_{\mathrm{H}} + 1\right)^2-\left(s_{\mathrm{H}} + 1\right)^3 = 2s_{\mathrm{H}}^3 + 3s_{\mathrm{H}}^2 - 1
\end{equation*}
and
\begin{equation}
    \cardSet{\edgeSet{\lambda c}}=c+1=2s_{\mathrm{H}}^3 + 3s_{\mathrm{H}}^2
\end{equation}
for the number of cotree elements on the wire basket. Then, the limit for the ratio $\nicefrac{n_{\lambda}}{n}$ is given by
\begin{equation}
    r_{\lambda}\coloneqq\lim\limits_{s_{\mathrm{H}}\rightarrow\infty}\frac{\cardSet{\edgeSet{\lambda c}}}{n}
    =
    \lim_{s_{\mathrm{H}}\rightarrow\infty}\left(\frac{2s_{\mathrm{H}}^{3} 
    + 
    3s_{\mathrm{H}}^{2}}{3s_{\mathrm{H}}^3s_{\mathrm{h}}(s_{\mathrm{h}}+1)^2}\right)
    =\frac{2}{3s_{\mathrm{h}}(s_{\mathrm{h}}+1)^2},
\end{equation}
where
\begin{equation*}
    n=N\cardSet{\edgeSet{}^{(i)}}=3s_{\mathrm{H}}^3s_{\mathrm{h}}(s_{\mathrm{h}}+1)^2
\end{equation*}
is valid.
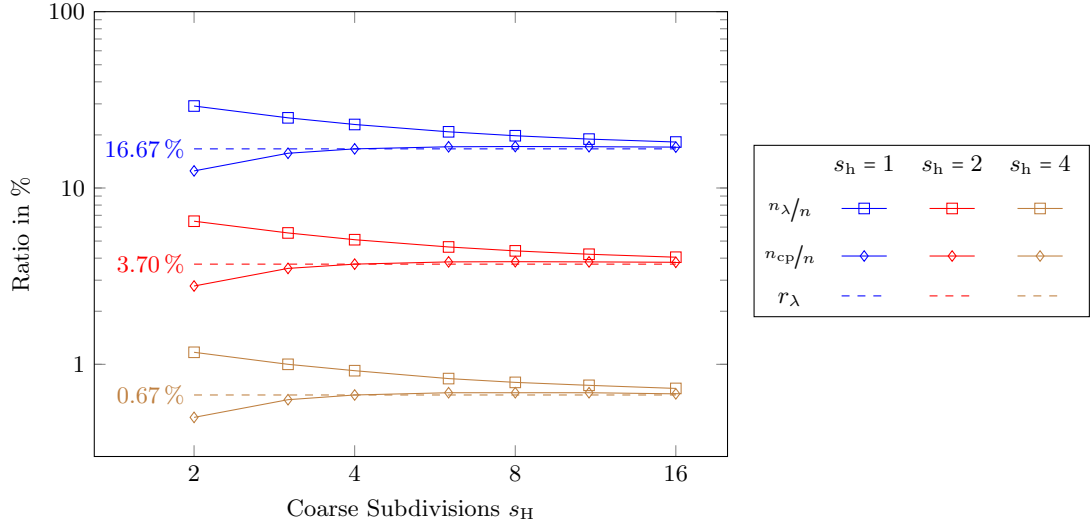
\begin{figure}
    \centering
    \begin{tikzpicture}
            \begin{loglogaxis}[
                xmin=1.3, xmax=20, ymin=0.3, ymax=100, ytick={1,10,100}, yticklabels={1,10,100}, xlabel={Coarse Subdivisions $s_{\mathrm{H}}$}, ylabel={Ratio in $\%$}, xtick={2,4,8,16}, xticklabels={2,4,8,16}, width=10cm, height=7.5cm, font=\small
            ]
                \addplot[mark=square, blue] table [
                    col sep=comma,
                    x = sH,
                    y expr = {ifthenelse(\thisrow{sh}==1,\thisrow{ratioWL},NaN)},
                ] {data/testFullTreeOnWireframe.csv};\label{plt:rWL1}
                \addplot[mark=square, red] table [
                    col sep=comma,
                    x = sH,
                    y expr = {ifthenelse(\thisrow{sh}==2,\thisrow{ratioWL},NaN)},
                ] {data/testFullTreeOnWireframe.csv};\label{plt:rWL2}
                \addplot[mark=square, brown] table [
                    col sep=comma,
                    x = sH,
                    y expr = {ifthenelse(\thisrow{sh}==4,\thisrow{ratioWL},NaN)},
                ] {data/testFullTreeOnWireframe.csv};\label{plt:rWL3}

                \addplot[mark=diamond, blue] table [
                    col sep=comma,
                    x = sH,
                    y expr = {ifthenelse(\thisrow{sh}==1,\thisrow{ratioPL},NaN)},
                ] {data/testFullTreeOnWireframe.csv};\label{plt:rPL1}
                \addplot[mark=diamond, red] table [
                    col sep=comma,
                    x = sH,
                    y expr = {ifthenelse(\thisrow{sh}==2,\thisrow{ratioPL},NaN)},
                ] {data/testFullTreeOnWireframe.csv};\label{plt:rPL2}
                \addplot[mark=diamond, brown] table [
                    col sep=comma,
                    x = sH,
                    y expr = {ifthenelse(\thisrow{sh}==4,\thisrow{ratioPL},NaN)},
                ] {data/testFullTreeOnWireframe.csv};\label{plt:rPL3}

            \addplot[mark=none, blue, dashed] table [
                col sep=comma,
                x = sH,
                y expr = {ifthenelse(\thisrow{sh}==1,\thisrow{estRatioL},NaN)},
            ] {data/testFullTreeOnWireframe.csv};\label{plt:rEL1}
            \addplot[mark=none, red, dashed] table [
                col sep=comma,
                x = sH,
                y expr = {ifthenelse(\thisrow{sh}==2,\thisrow{estRatioL},NaN)},
            ] {data/testFullTreeOnWireframe.csv};\label{plt:rEL2}
            \addplot[mark=none, brown, dashed] table [
                col sep=comma,
                x = sH,
                y expr = {ifthenelse(\thisrow{sh}==4,\thisrow{estRatioL},NaN)},
            ] {data/testFullTreeOnWireframe.csv};\label{plt:rEL3}
            \node[blue,anchor=east] at (axis cs: 2,16.67){16.67\,\%};
            \node[red,anchor=east] at (axis cs: 2,3.70){3.70\,\%};
            \node[brown,anchor=east] at (axis cs: 2,0.67){0.67\,\%};
        \end{loglogaxis}
        \matrix[
                matrix of nodes,
                draw, font=\small,
                inner sep=0.2em, draw, ampersand replacement={\&},
            ] at (11.3,3){
                \& $~s_{\mathrm{h}}=1~$ \& $~s_{\mathrm{h}}=2~$ \& $~s_{\mathrm{h}}=4~$ \\[0.5em]
                $\nicefrac{\cardSet{\edgeSet{\lambda c}}}{n}$ \& \ref{plt:rWL1} \& \ref{plt:rWL2} \& \ref{plt:rWL3} \\[0.5em]
                $\nicefrac{\nCoarse}{n}$ \& \ref{plt:rPL1} \& \ref{plt:rPL2} \& \ref{plt:rPL3} \\[0.5em]
                $r_{\lambda}$ \& \ref{plt:rEL1} \& \ref{plt:rEL2} \& \ref{plt:rEL3} \\[0.5em]
            };
            
    \end{tikzpicture}
    \caption{Ratio of cotree DOFs on wire basket $\cardSet{\edgeSet{\lambda c}}$ and primal DOFs $\nCoarse$ with respect to $n$ for varying $s_{\mathrm{h}}$ and $s_{\mathrm{H}}$.}
    \label{fig:ratios}
\end{figure}
\revised{Due to the relation $\cardSet{\edgeSet{\lambda c}}\geq \nCoarse$, we know that $\nicefrac{\nCoarse}{n}$ is bounded by the respective ratio with $\cardSet{\edgeSet{\lambda c}}$. Hence, $r_{\lambda}$ should be a good indicator for $\nicefrac{\nCoarse}{n}$ if $s_{\mathrm{H}}$ is large. The explained ratios are visualized in \autoref{fig:ratios} for different combinations of $s_{\mathrm{h}}$ and $s_{\mathrm{H}}$ where we can observe all expected properties. The estimated ratio $r_{\lambda}$ is close to $\nicefrac{\nCoarse}{n}$, therefore giving a good indication of the coarse problem size which can even be computed before assembling the discrete problem. We conclude that the coarse problem is only a fraction of the overall problem if the local systems are sufficiently refined.}